\documentclass[11pt]{article}
\usepackage[T2A]{fontenc}
\usepackage{ alphalph,etoolbox }
\usepackage{amsthm, amsmath, amssymb, amsbsy, amscd, amsfonts, latexsym, euscript,epsf, bm}
\newtheorem{thm}{Theorem}[section]
\usepackage{bbold, bm, stackrel}
\usepackage{epic,eepic}
\usepackage{texdraw}
\usepackage[dvips]{color}
\usepackage{color}
\usepackage{ dsfont }
\usepackage{graphicx}
\usepackage{exscale,relsize}
\usepackage{authblk}
\usepackage{url}
\usepackage{hyperref}
\usepackage{enumerate}
\usepackage{graphicx}
\newtheorem{proposition}[thm]{Proposition}

\newtheorem{theorem}[thm]{Theorem}

\theoremstyle{definition}

\DeclareMathOperator{\tr}{tr}

\title{Noncommutative Boussinesq and NLS type 2-  and  3-simplex maps}
\date{}

\author{S. Konstantinou-Rizos\thanks{skonstantin84@gmail.com} ~and A.~A. Kutuzova\thanks{a.kutuzova6@uniyar.ac.ru}}
\affil{Centre of Integrable Systems, P.G. Demidov Yaroslavl State University, Yaroslavl, Russia}

\patchcmd{\subequations}{\alph{equation}}{\alphalph{\value{equation}}}{}{}

\begin{document}

\maketitle

\begin{abstract}
  We construct noncommutative maps related to the Boussinesq and Nonlinear Schr\"odinger (NLS)
equations with their variables belonging to a noncommutative division ring. We show that the noncommutative Boussinesq type map satisfies the Yang--Baxter equation, and it can be squeezed down to a noncommutative
version of the Boussinesq lattice equation. Moreover, we show that the noncommutative NLS type map is a Zamolodchikov tetrahedron map.
\end{abstract}

\bigskip


\begin{quotation}
\noindent{\bf PACS numbers:}
02.30.Ik, 02.90.+p, 03.65.Fd.
\end{quotation}

\begin{quotation}
\noindent{\bf Mathematics Subject Classification 2020:}
35Q55, 16T25.
\end{quotation}

\begin{quotation}
\noindent{\bf Mathematics Subject Classification 2020:}
35Q55, 16T25.
\end{quotation}

\begin{quotation}
\noindent{\bf Keywords:} Noncommutative Yang--Baxter maps, noncommutative tetrahedron Zamolodchikov maps, noncommutative Boussinesq lattice system, noncommutative Darboux--B\"acklund transformations, noncommutative NLS type equations.
\end{quotation}

\allowdisplaybreaks

\section{Introduction}\label{intro}
The Boussinesq and the Nonlinear Schr\"odinger --- in both their continuous and discrete versions --- are undoubtedly two of the most popular equations in the modern theory of integrable systems. They are met in many applications, including hydrodynamics, soliton theory, the propagation of light in nonlinear optical fibers. Moreover, some of their solutions admit similar behaviour \cite{Clarkson},  and they share many integrable properties, such as having Lax representations, admitting Darboux and B\"acklund transformations, and are related to solutions of the Yang--Baxter equation (see, e.g., \cite{Sokor-2020-2, Sokor-Sasha, Sokor-Pap}).

The Yang--Baxter and the Zamolodchikov tetrahedron equations are the second and third, respectively, members of the family of $n$-simplex equations, which are fundamental equations of mathematical physics and are strictly connected to the theory of integrable systems. In particular, Yang--Baxter (2-simplex) and Zamolodchikov tetrahedron (3-simplex) maps, which are set-theoretical solutions to the Yang--Baxter and Zamolodchikov tetrahedron equations, respectively, are related to integrable systems via the symmetries of the latter \cite{Kassotakis-Tetrahedron, Pap-Tongas-Veselov}, Darboux--B\"acklund transformations \cite{Sokor-Sasha, Sokor-Pap, MPW}, and other transformations \cite{Kouloukas, Pavlos-Maciej-2, Sokor-Kouloukas}.

Furthermore, a great deal of attention has been paid to the construction and study of noncommutative $n$-simplex maps and their relation to noncommutative integrable systems (see, e.g., \cite{Doliwa-Kashaev, Kassotakis-2023, Kassotakis-Kouloukas-Maciej} and the references therein). This includes the construction of vector versions of NLS and sine-Gordon type Yang--Baxter maps \cite{Sokor-Sasha, MPW}, noncommutative KdV and NLS type Yang--Baxter and tetrahedron maps via Darboux transformations \cite{Giota-Miky, Sokor-2022, Sokor-Nikitina}, Grassmann extensions of NLS, KdV and Boussinesq type Yang--Baxter \cite{Giota-Miky, GKM, Sokor-Sasha-2016}, noncommutative 3D-compatible Yang--Baxter maps \cite{Doliwa-2014, Kassotakis-2023}, noncommutative versions of Yang--Baxter maps from the $\mathcal{H}$ and $\mathcal{F}$ lists \cite{Kassotakis-2023, Kassotakis-Kouloukas}, as well as Yang--Baxter maps related to noncommutative relativistic collisions \cite{Kassotakis-Kouloukas-Maciej}.

In this paper we extend to the noncommutative case some results of the works \cite{Sokor-2020, Sokor-2020-2}, namely we construct fully noncommutative versions of 2- and 3-simplex maps related to the Boussinesq and the Zamolodchikov tetrahedron equations. The proof of the fact that a map with noncommutative variables is an $n$-simplex map by straightforward substitution to the $n$-simplex equation is a very difficult task, even using packages of symbolic computation for simple rational $2$-simplex maps. Therefore, in order to prove that a noncommutative map satisfies the Yang--Baxter equation or the Zamolodchikov tetrahedron equation, we use matrix tri- and six-factorisation problems, respectively \cite{Kouloukas, Sokor-2022}.

In what follows, we list the main results of the paper.
\begin{itemize}
    \item Construction of a new commutative Boussinesq type Yang--Baxter map;
    \item construction of a noncommutative version of the Boussinesq Yang--Baxter map which appeared in \cite{Sokor-2020};
    \item proof that the noncommutative Boussinesq type map satisfies the Yang--Baxter equation;
    \item squeeze down of the noncommutative Boussinesq type Yang--Baxter map to an integrable noncommutative Boussinesq system of partial difference equations;
    \item the construction of a noncommutative map that follows from a Darboux transformation for the NLS equation via the local Yang--Baxter equation;
    \item proof that that obtained noncommutative NLS type map is a Zamolodchikov tetrahedron map.
\end{itemize}

The rest of the text is organised as follows. In the next section we provide the reader with all the necessary information in order to understand the results of this paper. In particular, we introduce the notation used throughout the text, we give the definitions of 2- and 3-simplex equations and explain the relation between 2- and 3-simplex maps (i.e. set-theoretical solutions of the Yang--Baxter and Zamolodchikov tetrahedron equations, respectively) and matrix refactorisation problems. In Section \ref{Sec-BSQ}, we employ the Lax matrix of the discrete Boussinesq lattice equation and, using the correspondence approach \cite{Igonin-Sokor}, we construct a new, noninvolutive, Boussinesq type 2-simplex (Yang--Baxter) map. Moreover, we construct a noncommutative version of the Boussinesq type Yang--Baxter map presented in \cite{Sokor-2020} with its variables belonging to a noncommutative division ring. We show that the latter maps can be squeezed down to a noncommutative version of the Boussinesq lattice system. Section \ref{sec-NLS} deals with the contruction of a noncommutative NLS type 3-simplex (tetrahedron Zamolodchikov) map. Specifically, using a Darboux transformation for the noncommutative NLS equation \cite{Sokor-PaulX}, we construct a noninvolutive NLS type map via the local Yang--Baxter equation. Then, we prove that this map is a 3-simplex map on a noncommutative division ring.

\section{Preliminaries}
\subsection{Notation}\label{notation}
We shall adopt the following notation:
\begin{itemize}
    \item By $\mathcal{X}$ we denote an arbitrary set, whereas by Latin italic letters (i.e. $x, y, u, v$ etc.) the elements of $\mathcal{X}$, with the exception of the `spectral parameter' which is denoted by the Greek letter $\lambda$.
    
    \item By $\mathcal{X}^n$ we denote the Cartesian product $\mathcal{X}^n=\underbrace{\mathcal{X}\times \mathcal{X}\times\dots\times \mathcal{X}}_{n}$.
    
    \item By $\mathfrak{R}$ we denote a noncommutative division ring, and its elements are denoted by bold italic Latin letters (i.e. $\bm{x}, \bm{y}, \bm{u}$ etc.). That is, $\mathfrak{R}$ is a ring with multiplicative identity $1$ where commutativity with respect to mutliplication is not assumed, and every nonzero element $\bm{x}$ has inverse $\bm{x}^{-1}$, i.e. $\bm{x}\bm{x}^{-1}=\bm{x}^{-1}\bm{x}=1$. The centre of a division ring will be denoted by $Z(\mathfrak{R})=\{a\in\mathfrak{R}:\forall\bm{x}\in\mathfrak{R},a\bm{x}=\bm{x}a\}$.
    
    \item Matrices will be denoted by Latin capital straight letters (i.e. ${\rm A}, {\rm B}, {\rm C}$) etc. Moreover, matrix operators are denoted by capital caligraphic letters (for example, $\mathcal{L}=D_x-\rm{U}$).

    \item For a $2 \times 2$ matrix ${\rm A}$, the notation ${\rm A}^n_{ij}$ means the $n \times n$ extension of matrix ${\rm A}$ with its elements lying on the intersection of the $i,j$ rows with the $i,j$ columns. The rest elements of the $n \times n$ extension are 1 or 0 using the following rule: All elements of rows and columns where the elements of the original $2 \times 2$ matrix ${\rm A}$ are located are 0. For instance, for the matrix ${\rm A}=\begin{pmatrix} 
a & b \\ 
c & d 
\end{pmatrix}$, we have ${\rm A}^3_{13}\begin{pmatrix} 
 a & 0 & b\\ 
0 & 1 & 0\\
c & 0 & d
\end{pmatrix}$.

    \item Let $u$ be a function of two discrete variables $n,m\in\mathbb{Z}$, i.e. $u=u(n,m)$. By indices $u_{i,j}$ we denote the shifts with respect to $n$ and  $m$: $u_{i,j}=u(n+i,m+j)$.
\end{itemize}

\subsection{Set-theoretical Yang--Baxter equation and Lax representations}
Let $\mathcal{X}$ be a set. A map $Y:\mathcal{X}^2\rightarrow \mathcal{X}^2$ is called a \textit{Yang--Baxter map} if it satisfies the Yang--Baxter equation
\begin{equation}\label{eq_YB}
Y^{12}\circ Y^{13}\circ Y^{23}=Y^{23}\circ Y^{13}\circ Y^{12}.
\end{equation}
The terms $Y^{12}$, $Y^{13}$, $Y^{23}$ in \eqref{eq_YB} are maps $\mathcal{X}^3\to \mathcal{X}^3$ defined as follows 
$$
Y^{12}(x,y,z)=\big(u(x,y),v(x,y),z\big),~~~
Y^{23}(x,y,z)=\big(x,u(y,z),v(y,z)\big),~~~
Y^{13}(x,y,z)=\big(u(x,z),y,v(x,z)\big),
$$
where $x,y,z\in \mathcal{X}$.

Now, let $\mathbb{C}$ be the field of complex numbers. 
A \emph{parametric Yang--Baxter map} is a Yang--Baxter map with its variables assigned with two complex variables $a,b\in\mathbb{C}$. Specifically, it is a map
$$
Y_{a,b}\colon (\mathcal{X}\times\mathbb{C})\times (\mathcal{X}\times\mathbb{C})\to (\mathcal{X}\times\mathbb{C})\times (\mathcal{X}\times\mathbb{C}),
$$
namely,
$$
Y_{a,b}((x,a),(y,b))=((u(x,a),(y,b),a),(v(x,a),(y,b),b)),
$$
which we usually write in a more compact form as
\begin{gather}
\label{ParamYB}
Y_{a,b}(x,y)=\big(u_{a,b}(x,y),\,v_{a,b}(x,y)\big),\quad x,y\in \mathcal{X},
\quad a,b\in\mathbb{C},
\end{gather}
satisfying the parametric Yang--Baxter equation
\begin{gather}
\label{pybeq}
Y^{12}_{a,b}\circ Y^{13}_{a,c} \circ Y^{23}_{b,c}=
Y^{23}_{b,c}\circ Y^{13}_{a,c} \circ Y^{12}_{a,b}\quad\text{for all }\,a,b,c\in\mathbb{C}.
\end{gather}

The most celebrated parametric Yang--Baxter map is the Adler map \cite{Adler}
$$
Y:(x,y)\rightarrow \left(y-\frac{a-b}{x+y}, x+\frac{a-b}{x+y}\right)
$$
which is related to the discrete potential KdV equation
\begin{equation}\label{pkdv}
(f_{11}-f)(f_{01}-f_{10})=a-b,
\end{equation}
via the symmetries of  the latter \cite{Pap-Tongas-Veselov}.

Now, let ${\rm L}(x;a,\lambda)$, $x\in\mathcal{X},a\in\mathbb{C}$, be a square matrix and $Y_{a,b}(x,y)=\big(u_{a,b}(x,y),\,v_{a,b}(x,y)\big)$ a parametric Yang--Baxter map. Suppose that $u=u_{a,b}(x,y)$ and $v=v_{a,b}(x,y)$ obey the equation
\begin{gather}
\label{eq-Lax}
{\rm L}(u,a,\lambda){\rm L}(v,b,\lambda)={\rm L}(y,b,\lambda){\rm L}(x,a,\lambda)
\end{gather}
for all values of $x,y,a,b,\lambda$. 
Then, matrix ${\rm L}(x;a,\lambda)$ 
is called a \emph{Lax matrix} for $Y_{a,b}$ \cite{Veselov2}, and the matrix refactorisation problem \eqref{eq-Lax} is called its Lax representation \cite{Dimakis-Hoissen-2015}. 

It is important to note that not every map satisfying \eqref{eq-Lax} is a parametric Yang--Baxter map. In particular, we have the following.

\begin{theorem} (Kouloukas--Papageorgiou \cite{Kouloukas})\label{KP}
Let $Y_{a,b}$ be a map with Lax representation \eqref{eq-Lax}. If the following \textit{matrix trifactorisation problem}
\begin{gather}
\label{trifac}
{\rm L}(u,a,\lambda){\rm L}(v,b,\lambda){\rm L}(w,c,\lambda)={\rm L}(z,c,\lambda){\rm L}(y,b,\lambda){\rm L}(x,a,\lambda),\quad\text{for all}~~a,b,c\in\mathbb{C}
\end{gather}
implies $u=x$, $v=y$, $w=z$, then
$Y_{a,b}$ satisfies the parametric Yang--Baxter equation \eqref{ParamYB}.
\end{theorem}

\subsection{Zamolodchikov tetrahedron and local Yang--Baxter equation}
A map $T:\mathcal{X}^3\rightarrow \mathcal{X}^3$, i.e. $T(x,y,z)=(u(x,y,z),v(x,y,z),w(x,y,z))$, $x,y,z\in\mathcal{X}$,
is called a Zamolodchikov tetrahedron (3-simplex) map if it solves the equation
\begin{equation}\label{Tetrahedron-eq}
    T^{123}\circ T^{145} \circ T^{246}\circ T^{356}=T^{356}\circ T^{246}\circ T^{145}\circ T^{123}.
\end{equation}
This is the 3-simplex or Zamolodhikov tetrahedron equation.
The functions $T^{ijk}:\mathcal{X}^6\rightarrow  \mathcal{X}^6$, $i,j=1,2,3,~i\neq j$, act on the terms $ijk$ of the product $\mathcal{X}^6$ as map $T$ and trivially on the others. For instance,
$$
T^{356}(x,y,z,r,s,t)=(x,y,u(z,s,t),r,v(z,s,t),w(z,s,t)).
$$

Let $a$, $b$ and $c$ be complex parameters. A map $T:(\mathcal{X}\times\mathbb{C})^3\rightarrow (\mathcal{X}\times\mathbb{C})^3$, namely $T:((x,a),(y,b),(z,c))\mapsto ((u((x,a),(y,b),(z,c)),a),(v((x,a),(y,b),(z,c)),b),(w((x,a),(y,b),(z,c)),c))$, is called parametric 3-simplex or tetrahedron map if it solves the parametric 3-simplex equation
\begin{equation}\label{Par-Tetrahedron-eq}
    T^{123}_{a,b,c}\circ T^{145}_{a,d,e} \circ T^{246}_{b,d,f}\circ T^{356}_{c,e,f}=T^{356}_{c,e,f}\circ T^{246}_{b,d,f}\circ T^{145}_{a,d,e}\circ T^{123}_{a,b,c}.
\end{equation}
Parametric 3-simplex maps will be denoted in short as:
\begin{equation}\label{Par-Tetrahedron_map}
 T_{a,b,c}:(x,y,z)\mapsto (u_{a,b,c}(x,y,z),v_{a,b,c}(x,y,z),w_{a,b,c}(x,y,z)),
\end{equation}

Now, let ${\rm L}={\rm L}(x,k,\lambda)$, $x\in\mathcal{X}$,  $k,\lambda\in\mathbb{C}$, be a matrix of the form
\begin{equation}\label{matrix-L}
   {\rm L}(x,k,\lambda)= \begin{pmatrix} 
a_{11}(x,k,\lambda) & a_{12}(x,k,\lambda)\\ 
a_{21}(x,k,\lambda) & a_{22}(x,k,\lambda)
\end{pmatrix},
\end{equation}
for scalar functions $a_{ij}$. Consider the $3\times 3$ and $4\times 4$ extensions of matrix \eqref{matrix-L}, given by
\begin{align}\label{Lij-mat}
   &{\rm L}^3_{12}(x,k,\lambda)=\begin{pmatrix} 
 a_{11}(x,k,\lambda) &  a_{12}(x,k,\lambda) & 0\\ 
a_{21}(x,k,\lambda) &  a_{22}(x,k,\lambda) & 0\\
0 & 0 & 1
\end{pmatrix},\nonumber\\
& {\rm L}^3_{13}(x,k,\lambda)= \begin{pmatrix} 
 a_{11}(x,k,\lambda) & 0 & a_{12}(x,k,\lambda)\\ 
0 & 1 & 0\\
a_{21}(x,k,\lambda) & 0 & a_{22}(x,k,\lambda)
\end{pmatrix},  \\
 &{\rm L}^3_{23}(x,k,\lambda)=\begin{pmatrix} 
   1 & 0 & 0 \\
0 & a_{11}(x,k,\lambda) & a_{12}(x,k,\lambda)\\ 
0 & a_{21}(x,k,\lambda) & a_{22}(x,k,\lambda)
\end{pmatrix},\nonumber
\end{align}
and 
{\small
\begin{align}\label{Lij4}
        &{\rm L}^4_{12}(x,k,\lambda)=\begin{pmatrix}a_{11}(x,k,\lambda) & a_{12}(x,k,\lambda) & 0 & 0 \\ a_{21}(x,k,\lambda) & a_{22}(x,k,\lambda) & 0 & 0 \\ 0 & 0 & 1 & 0 \\ 0 & 0 & 0 & 1\end{pmatrix},\quad
    {\rm L}^4_{13}(x,k,\lambda)=\begin{pmatrix}a_{11}(x,k,\lambda) & 0 & a_{12}(x,k,\lambda) & 0 \\ 0 & 1 & 0 & 0 \\ a_{21}(x,k,\lambda) & 0 & a_{22}(x,k,\lambda) & 0 \\ 0 & 0 & 0 & 1\end{pmatrix},\nonumber\\
        &{\rm L}^4_{23}(x,k,\lambda)=\begin{pmatrix}1 & 0 & 0 & 0 \\ 0 & a_{11}(x,k,\lambda) & a_{12}(x,k,\lambda) & 0 \\ 0 & a_{21}(x,k,\lambda) & a_{22}(x,k,\lambda) & 0 \\ 0 & 0 & 0 & 1\end{pmatrix},\quad {\rm L}^4_{14}(x,k,\lambda)=\begin{pmatrix}a_{11}(x,k,\lambda) & 0 & 0 & a_{12}(x,k,\lambda) \\ 0 & 1 & 0 & 0 \\ 0 & 0 & 1 & 0 \\ a_{21}(x,k,\lambda) & 0 & 0 & a_{22}(x,k,\lambda)\end{pmatrix}, \\
        & {\rm L}^4_{24}(x,k,\lambda)=\begin{pmatrix}1 & 0 & 0 & 0 \\ 0 & a_{11}(x,k,\lambda) & 0 & a_{12}(x,k,\lambda) \\ 0 & 0 & 1 & 0 \\ 0 & a_{21}(x,k,\lambda) & 0 & a_{22}(x,k,\lambda)\end{pmatrix}, \quad
    {\rm L}^4_{34}(x,k,\lambda)=\begin{pmatrix}1 & 0 & 0 & 0 \\ 0 & 1 & 0 & 0 \\ 0 & 0 & a_{11}(x,k,\lambda) & a_{12}(x,k,\lambda) \\ 0 & 0 & a_{21}(x,k,\lambda) & a_{22}(x,k,\lambda)\end{pmatrix},\nonumber
\end{align}}
respectively.

Then, ${\rm L}={\rm L}(x,k,\lambda)$ is a Lax matrix for map \eqref{Par-Tetrahedron_map}, if the latter satisfies \cite{Dimakis-Hoissen-2015}
\begin{equation}\label{Lax-Tetra}
    {\rm L}^3_{12}(u,a,\lambda){\rm L}^3_{13}(v,b,\lambda){\rm L}^3_{23}(w,c,\lambda)= {\rm L}^3_{23}(z,c,\lambda){\rm L}^3_{13}(y,b,\lambda){\rm L}^3_{12}(x,a,\lambda),\quad \text{for any}\quad\lambda\in\mathbb{C}.
\end{equation}

Equation \eqref{Lax-Tetra} is called the \textit{local Yang--Baxter} equation or Maillet--Nijhoff equation \cite{Nijhoff} in Korepanov's form. In this matrix form equation \eqref{Lax-Tetra} was suggested by Korepanov, and deserves to be referred as \textit{Korepanov's equation}. The Korepanov equation allows to construct tetrahedron maps related to integrable systems of mathematical  physics \cite{Kassotakis-Tetrahedron, Sokor-2022} and `classify' tetrahedron maps \cite{Kashaev-Sergeev, Sergeev}. Note that, here, the entries $a_{ij}$ of matrix \eqref{matrix-L} are scalar, however they can  be matrices \cite{Korepanov}).

If a map \eqref{Par-Tetrahedron_map} satisfies the local Yang--Baxter equation \eqref{Lax-Tetra}, then this map is a 3-simplex map if some supplementary condition is satisfied. Specifically, we have the following.

\begin{theorem}\cite{Sokor-2022}\label{six-factorisation}
Let map \eqref{Par-Tetrahedron_map} be a solution of the local Yang--Baxter equation \eqref{Lax-Tetra} for some matrix ${\rm L}={\rm L}(x,a)$. Then, if equation
\begin{align}\label{6-fac}
&{\rm L}^4_{34}(\hat{t},a_6){\rm L}^4_{24}(\hat{s},a_5){\rm L}^4_{14}(\hat{r},a_4){\rm L}^4_{23}(\hat{z},a_3){\rm L}^4_{13}(\hat{y},a_2){\rm L}^4_{12}(\hat{x},a_1)=\nonumber\\
&{\rm L}^4_{34}(t,a_6){\rm L}^4_{24}(s,a_5){\rm L}^4_{14}(r,a_4){\rm L}^4_{23}(z,a_3){\rm L}^4_{13}(y,a_2){\rm L}^4_{12}(x,a_1)
\end{align}
implies the trivial solution $\hat{t}=t$, $\hat{s}=s$, $\hat{r}=r$, $\hat{z}=z$, $\hat{y}=y$ and $\hat{x}=x$, then map \eqref{Par-Tetrahedron_map} is a parametric tetrahedron map.
\end{theorem}

\section{Boussinesq type 2-simplex maps}
Consider the Boussinesq matrix \cite{Tongas-Nijhoff}:
\begin{equation}\label{Lax-BSQ}
{\rm L}(p,q,q_{10},r_{10},a):=
\left(
\begin{matrix}
 -q_{10} & 1 & 0 \\
 -r_{10} & 0 & 1 \\
a-p q_{10}-q r_{10}-\lambda & p & q
\end{matrix}\right),
\end{equation}
which is the Lax matrix of the lattice Boussinesq system. Here, $p$, $q$ and $r$ are the Boussinesq potentials which are functions of two discrete variables $n,m\in\mathbb{Z}$. That is, using the notation introduced in section \ref{notation}, $p_{10}=p(n+1,m)$, $q_{10}=q(n+1,m)$ and $r_{10}=r(n+1,m)$.

We shall construct commutative and noncommutative Yang--Baxter maps $Y:\mathcal{X}\rightarrow\mathcal{X}$ generated by the Boussinesq Lax matrix \eqref{Lax-BSQ}.

\subsection{Commutative Boussinesq type Yang--Baxter map}\label{Sec-BSQ}
Let $\mathcal{X}=\mathbb{C}$. We replace the elements of \eqref{Lax-BSQ} by variables $x_i\in\mathcal{X}$, $i=1,2,3,4$, namely:
$$
(p,q,q_{10},r_{10})\rightarrow (\bm{x}_1, \bm{x}_2, \bm{x}_3, \bm{x}_4),
$$
that is, we consider the  matrix:
\begin{equation}\label{C-Lax-BSQ}
{\rm L}(\bm{x}_1, \bm{x}_2, \bm{x}_3, \bm{x}_4,a):=
\left(
\begin{matrix}
 -\bm{x}_3 & 1 & 0 \\
 -\bm{x}_4 & 0 & 1 \\
a-\bm{x}_1 \bm{x}_3-\bm{x}_2 \bm{x}_4-\lambda & \bm{x}_1 & \bm{x}_2
\end{matrix}\right).
\end{equation}

Consider the matrix refactorisation problem:
\begin{equation}
    {\rm L}(u_1, u_2, u_3, u_4,a){\rm L}(v_1, v_2,v_3, v_4,b)={\rm L}(y_1, y_2, y_3, y_4,b){\rm L}(x_1, x_2, x_3, x_4,a).
\end{equation}
This equation is equivalent to a system of polynomial equations  
$$
P_i(u_1, u_2, u_3, u_4,v_1, v_2,v_3, v_4,x_1, x_2, x_3, x_4,y_1, y_2, y_3, y_4;a,b)=0,\quad i=1,\ldots 7,
$$
which can be solved for $u_1,u_2,u_3,u_4,v_2,v_3$ and $v_4$ in terms of $v_1$, namely it is equivalent to
\begin{minipage}{.5\linewidth}
\begin{align*}
  u_1&=y_1+\frac{a-b}{x_1-y_4+x_2y_3}x_2,\\
	u_2&=y_2+\frac{b-a}{x_1-y_4+x_2y_3},\\
	u_3&=y_3,\\
	u_4&=y_4+v_1-x_1,
\end{align*}
\end{minipage}
\hspace{-.9 cm}
\begin{minipage}{.5\linewidth}
\begin{align}
  v_2&=x_2,\nonumber\\
	v_3&=x_3+\frac{b-a}{x_1-y_4+x_2y_3},\label{correspondence-BSQ}\\
	v_4&=x_4+\frac{b-a}{x_1-y_4+x_2y_3}y_3.\nonumber
\end{align}
\end{minipage}\\

In order for the above system to define a map, one needs to supplement it with one more equation for $v_1$. In \cite{Sokor-2020}, we constructed a Boussinesq type Yang--Baxter map (see Proposition 3.2.1 in \cite{Sokor-2020})  for the choice $v_1=x_1$. Here, we construct a new solution to the Yang--Baxter equation. In particular, we have the following.

\begin{proposition}
    Map $Y_1:\mathbb{C}^8\rightarrow \mathbb{C}^8$, $i=1,2$, given by
    
 \begin{minipage}{.5\linewidth}
\begin{align*}
 x_1\mapsto u_1&=y_1+\frac{a-b}{x_1+x_2y_3-y_4}x_2,\\
 x_2\mapsto u_2&=y_2-\frac{a-b}{x_1+x_2y_3-y_4},\\
 x_3\mapsto u_3&=y_3,\\
 x_4\mapsto u_4&=y_4+\frac{a-b}{x_1+x_2y_3-y_4}x_2,
\end{align*}
\end{minipage}
\hspace{-.9 cm}
\begin{minipage}{.5\linewidth}
\begin{align}\label{YB-2-BSQ}
y_1\mapsto v_1&=x_1-\frac{a-b}{x_1+x_2y_3-y_4}x_2, \nonumber\\
y_2\mapsto v_2&=x_2,\nonumber\\
y_3\mapsto v_3&=x_3-\frac{a-b}{x_1+x_2y_3-y_4},\\
y_4\mapsto v_4&=x_4-\frac{a-b}{x_1+x_2y_3-y_4}y_3,\nonumber
\end{align}
\end{minipage}
\vspace{.2cm}\\ 
is an eight-dimensional parametric Yang--Baxter map, and it admits the following functionally independent first integrals
\begin{subequations}\label{BSQ-ints}
\begin{align}
I_1&=x_1+y_1,\label{BSQ-ints-a}\\
I_2&=x_2+y_2-x_3-y_3,\label{BSQ-ints-b}\\
I_3&=x_2y_2+x_3y_3-x_4-y_4,\label{BSQ-ints-c}\\
I_4&=b(x_2-x_3)-a(y_3-y_2)+(x_4-x_3y_2-y_1)(x_1+x_2y_3-y_4).\label{BSQ-ints-d}
\end{align}
\end{subequations}
\end{proposition}
\begin{proof}
    If we supplement system \eqref{correspondence-BSQ} with equation 
    $$
    u_1+v_1=y_1+x_1,
    $$
    then the associated augmented system is equivalent to map \eqref{YB-2-BSQ}. It can be shown by straightforward substitution to the Yang--Baxter equation that map \eqref{YB-2-BSQ} is a parametric Yang--Baxter map.

    Integrals $I_1$ and $I_3$ are found from the trace of the monodromy matrix:
    $$
    \tr({\rm L}(y_1, y_2, y_3, y_4,b){\rm L}(x_1, x_2, x_3, x_4,a))=I_1+I_3.
    $$
    Integrals $I_2$ and $I_4$ are found from the determinant $\det({\rm L}(y_1, y_2, y_3, y_4,b){\rm L}(x_1, x_2, x_3, x_4,a)- k\cdot  \mathbb{I}_3)$, where $\mathbb{I}_3$ is the identity matrix $3\times 3$. Now, the rank of matrix $r=\left[\nabla I_1,\nabla I_2,\nabla I_3, \nabla I_4\right]$ is 4, thus $I_i$, $i=1,2,3,4$ are functionally independent.
\end{proof}

\subsection{Nonommutative Boussinesq type Yang--Baxter maps}
Now, let $\mathcal{X}=\mathcal{R}$, that is, a non-commutative division ring. We replace the elements of \eqref{Lax-BSQ} by arbitrary elements of $\mathcal{R}$, $\bm{x}_i\in\mathcal{R}$, $a\in  Z(\mathcal{R})$,
$$
(p,q,q_{10},r_{10})\rightarrow (\bm{x}_1, \bm{x}_2, \bm{x}_3, \bm{x}_4)
$$
namely, we consider the matrix
\begin{equation}\label{NC-Lax-BSQ}
{\rm L}(\bm{x}_1, \bm{x}_2, \bm{x}_3, \bm{x}_4,a):=
\left(
\begin{matrix}
 -\bm{x}_3 & 1 & 0 \\
 -\bm{x}_4 & 0 & 1 \\
a-\bm{x}_1 \bm{x}_3-\bm{x}_2 \bm{x}_4-\lambda & \bm{x}_1 & \bm{x}_2
\end{matrix}\right),
\end{equation}
and substitute it to the following matrix refactorisation problem:
\begin{equation}
    {\rm L}(\bm{u}_1, \bm{u}_2, \bm{u}_3, \bm{u}_4,a){\rm L}(\bm{v}_1, \bm{v}_2, \bm{v}_3, \bm{v}_4,b)={\rm L}(\bm{y}_1, \bm{y}_2, \bm{y}_3, \bm{y}_4,b){\rm L}(\bm{x}_1, \bm{x}_2, \bm{x}_3, \bm{x}_4,a).
\end{equation}
The above equation is equivalent to $\bm{u}_3=\bm{y}_3$, $\bm{v}_2=\bm{x}_2$ the following system of polynomial equations
\begin{align*}
&\bm{y}_3\bm{v}_3-\bm{v}_4=\bm{y}_3\bm{x}_3-\bm{x}_4,\\
&\bm{u}_1+\bm{u}_2\bm{x}_2=\bm{y}_1+\bm{y}_2\bm{x}_2,\\
&\bm{u}_4-\bm{v}_1=\bm{y}_4-\bm{x}_1,\\
&(\bm{u}_4-\bm{v}_1)\bm{v}_3+b-\bm{x}_2\bm{v}_4=(\bm{y}_4-\bm{x}_1)\bm{x}_3+a-\bm{x}_2\bm{x}_4,\\
&\bm{u}_1\bm{v}_4+(a-\bm{u}_1\bm{y}_3-\bm{u}_2\bm{u}_4)\bm{v}_3-\bm{u}_2(b-\bm{v}_1\bm{v}_3-\bm{x}_2 \bm{v}_4)=\\
&\bm{y}_1\bm{x}_4+(b-\bm{y}_1\bm{y}_3-\bm{y}_2\bm{y}_4)\bm{x}_3-\bm{y}_2(a-\bm{x}_1\bm{x}_3-\bm{x}_2 \bm{x}_4),
\end{align*}
for the rest variables $\bm{u}_1, \bm{u}_2, \bm{u}_4, \bm{v}_1, \bm{v}_3$ and $\bm{v}_4$.

This system can be solved for $\bm{u}_1, \bm{u}_2, \bm{u}_3, \bm{u}_4, \bm{v}_2, \bm{v}_3$ and $\bm{v}_4$ in terms of $\bm{v}_1$, i.e. it is equivalent to
\begin{align}\label{NC-BSQ-corr}
&\bm{u}_1=\bm{y}_1-(a-b) (\bm{y}_4-\bm{x}_1-\bm{x}_2\bm{y}_3)^{-1}\bm{x}_2,\nonumber\\
&\bm{u}_2=\bm{y}_2+(a-b)(\bm{y}_4-\bm{x}_1-\bm{x}_2\bm{y}_3)^{-1},\nonumber\\
&\bm{u}_3=\bm{y}_3,\nonumber\\
&\bm{u}_4=\bm{y}_4+\bm{v}_1-\bm{x}_1,\\
&\bm{v}_2=\bm{x}_2,\nonumber\\
&\bm{v}_3=(a-b)(\bm{y}_4-\bm{x}_1-\bm{x}_2\bm{y}_3)^{-1}+\bm{x}_3.\nonumber\\
&\bm{v}_4=\bm{x}_4+(a-b)\bm{y}_3(\bm{y}_4-\bm{x}_1-\bm{x}_2\bm{y}_3)^{-1}.\nonumber
\end{align}
The equations of the above system are not enough to define a map, thus \eqref{NC-BSQ-corr} defines a correspondence between $\mathcal{R}^8$ and $\mathcal{R}^8$. In order to define a map, we need to supplement system \eqref{NC-BSQ-corr} with one equation. In particular, we have the following.

\begin{proposition}(Noncommutative Boussinesq type map)
    The above correspondence defines an eight-dimensional noninvolutive map $Y_2:\mathcal{R}^8\rightarrow \mathcal{R}^8$ given by:

 \begin{minipage}{.5\linewidth}
\begin{align*}
&\bm{x}_1\mapsto\bm{u}_1=\bm{y}_1-(a-b) (\bm{y}_4-\bm{x}_1-\bm{x}_2\bm{y}_3)^{-1}\bm{x}_2,\\
&\bm{x}_2\mapsto\bm{u}_2=\bm{y}_2+(a-b)(\bm{y}_4-\bm{x}_1-\bm{x}_2\bm{y}_3)^{-1},\\
&\bm{x}_3\mapsto\bm{u}_3=\bm{y}_3,\\
&\bm{x}_4\mapsto\bm{u}_4=\bm{y}_4,
\end{align*}
\end{minipage}
\hspace{-.9 cm}
\begin{minipage}{.5\linewidth}
\begin{align}
&\bm{y}_2\mapsto\bm{v}_1=\bm{x}_1,\nonumber\\
&\bm{y}_2\mapsto\bm{v}_2=\bm{x}_2,\label{NC-BSQ-map}\\
&\bm{y}_3\mapsto\bm{v}_3=\bm{x}_3+(a-b)(\bm{y}_4-\bm{x}_1-\bm{x}_2\bm{y}_3)^{-1},\nonumber\\
&\bm{y}_4\mapsto\bm{v}_4=\bm{x}_4+(a-b)\bm{y}_3(\bm{y}_4-\bm{x}_1-\bm{x}_2\bm{y}_3)^{-1},\nonumber
\end{align}
\end{minipage}
\vspace{.2cm}\\ 
and it admits the following invariant
$$
I={\bm x}_2+{\bm y}_2-{\bm x}_3-{\bm y}_3.
$$
\end{proposition}
\begin{proof}
     If we supplement system \eqref{NC-BSQ-corr} with equation $\bm{u}_1=\bm{x}_1$, then the associated augmented system is equivalent to map \eqref{NC-BSQ-map}. Moreover, $I\circ Y_2= I$, i.e. $I$ is a first integral. Finally, we have that $\bm{u}_4\circ Y_2 \stackrel{\eqref{NC-BSQ-map}}{=}\bm{x}_4+(a-b)\bm{y}_3(\bm{y}_4-\bm{x}_1-\bm{x}_2\bm{y}_3)^{-1}\neq \bm{x}_4$, therefore map \eqref{NC-BSQ-map} is noninvolutive.
\end{proof}

For the Yang--Baxter property, we have the following.

\begin{theorem}
    Map \eqref{NC-BSQ-map} is a Yang--Baxter map.
\end{theorem}
\begin{proof}
We employ matrix $L(\bm{x}_1,\bm{x}_2,\bm{x}_3,\bm{x}_4,a)$ given by \eqref{NC-Lax-BSQ}. Using the left-hand side of the Yang--Baxter equation, and taking into account \eqref{NC-BSQ-map}, we have:
    \begin{align}
        &{\rm L}(\bm{z}_1,\bm{z}_2,\bm{z}_3,\bm{z}_4,c){\rm L}(\bm{y}_1,\bm{y}_2,\bm{y}_3,\bm{y}_4,b)L(\bm{x}_1,\bm{x}_2,\bm{x}_3,\bm{x}_4,a)=\nonumber\\
        &{\rm L}(\bm{z}_1,\bm{z}_2,\bm{z}_3,\bm{z}_4,c){\rm L}(\tilde{\bm{x}}_1,\tilde{\bm{x}}_2,\bm{y}_3,\bm{y}_4,a){\rm L}(\bm{x}_1,\bm{x}_2,\tilde{\bm{y}}_3,\tilde{\bm{y}}_4,b)=\nonumber\\
        &{\rm L}(\tilde{\tilde{\bm{x}}}_1,\tilde{\tilde{\bm{x}}}_2,\bm{z}_3,\bm{z}_4,a){\rm L}(\tilde{\bm{x}}_1,\tilde{\bm{x}}_2,\tilde{\bm{z}}_3,\tilde{\bm{z}}_4,c){\rm L}(\bm{x}_1,\bm{x}_2,\tilde{\bm{y}}_3,\tilde{\bm{y}}_4,b)=\label{LHS-BSQ}\\
        &{\rm L}(\tilde{\tilde{\bm{x}}}_1,\tilde{\tilde{\bm{x}}}_2,\bm{z}_3,\bm{z}_4,a){\rm L}(\tilde{\bm{x}}_1,\tilde{\bm{x}}_2,\tilde{\bm{z}}_3,\tilde{\bm{z}}_4,b){\rm L}(\bm{x}_1,\bm{x}_2,\tilde{\tilde{\bm{z}}}_3,\tilde{\tilde{\bm{z}}}_4,c).\nonumber
    \end{align}
Moreover, using the right-hand side of the Yang--Baxter equation, and taking into account \eqref{NC-BSQ-map}, it follows that  
\begin{align}
        &{\rm L}(\bm{z}_1,\bm{z}_2,\bm{z}_3,\bm{z}_4,c){\rm L}(\bm{y}_1,\bm{y}_2,\bm{y}_3,\bm{y}_4,b){\rm L}(\bm{x}_1,\bm{x}_2,\bm{x}_3,\bm{x}_4,a)=\nonumber\\
        &{\rm L}(\hat{\bm{y}}_1,\hat{\bm{y}}_2,\bm{z}_3,\bm{z}_4,b){\rm L}(\bm{y}_1,\bm{y}_2,\hat{\bm{z}}_3,\hat{\bm{z}}_4,c){\rm L}(\bm{x}_1,\bm{x}_2,\bm{x}_3,\bm{x}_4,a)=\nonumber\\
        &{\rm L}(\hat{\bm{y}}_1,\hat{\bm{y}}_2,\bm{z}_3,\bm{z}_4,b){\rm L}(\hat{\bm{x}}_1,\hat{\bm{x}}_2,\hat{\bm{z}}_3,\hat{\bm{z}}_4,a){\rm L}(\bm{x}_1,\bm{x}_2,\hat{\hat{\bm{z}}}_3,\hat{\hat{\bm{z}}}_4,c)=\label{RHS-BSQ}\\
        &{\rm L}(\hat{\hat{\bm{x}}}_1,\hat{\hat{\bm{x}}}_2,\bm{z}_3,\bm{z}_4,a){\rm L}(\hat{\bm{x}}_1,\hat{\bm{x}}_2,\hat{\bm{z}}_3,\hat{\bm{z}}_4,b){\rm L}(\bm{x}_1,\bm{x}_2,\hat{\hat{\bm{z}}}_3,\hat{\hat{\bm{z}}}_4,c).\nonumber
    \end{align}
Equations \eqref{LHS-BSQ} and \eqref{RHS-BSQ} imply  
\begin{align*}
    &{\rm L}(\tilde{\tilde{\bm{x}}}_1,\tilde{\tilde{\bm{x}}}_2,\bm{z}_3,\bm{z}_4,a){\rm L}(\tilde{\bm{x}}_1,\tilde{\bm{x}}_2,\tilde{\bm{z}}_3,\tilde{\bm{z}}_4,b){\rm L}(\bm{x}_1,\bm{x}_2,\tilde{\tilde{\bm{z}}}_3,\tilde{\tilde{\bm{z}}}_4,c)=\\
    &{\rm L}(\hat{\hat{\bm{x}}}_1,\hat{\hat{\bm{x}}}_2,\bm{z}_3,\bm{z}_4,a){\rm L}(\hat{\bm{x}}_1,\hat{\bm{x}}_2,\hat{\bm{z}}_3,\hat{\bm{z}}_4,b){\rm L}(\bm{x}_1,\bm{x}_2,\hat{\hat{\bm{z}}}_3,\hat{\hat{\bm{z}}}_4,c).
\end{align*}

For simplicity of the notation, we rename variables $\tilde{\tilde{x}}_1=z_1$,~$\tilde{\tilde{x}}_2=z_2$,~ $\tilde{x}_1=y_1$,~ $\tilde{x}_2=y_2$,~ $\tilde{z}_3=y_3$,~ $\tilde{z}_4=y_4$. Then, the above equation can be rewritten as
\begin{align}
    &{\rm L}(\bm{z}_1,\bm{z}_2,\bm{z}_3,\bm{z}_4,a){\rm L}(\bm{y}_1,\bm{y}_2,\bm{y}_3,\bm{y}_4,b){\rm L}(\bm{x}_1,\bm{x}_2,\bm{x}_3,\bm{x}_4,c)=\nonumber\\
    &{\rm L}(\bm{w}_1,\bm{w}_2,\bm{z}_3,\bm{z}_4,a){\rm L}(\bm{v}_1,\bm{v}_2,\bm{v}_3,\bm{v}_4,b){\rm L}(\bm{x}_1,\bm{x}_2,\bm{u}_3,\bm{u}_4,c).\label{trifac-BSQ}
\end{align}
This equation is equivalent to the system:
\begin{subequations}
\begin{align}
       & (-z_3v_3+v_4)u_3+z_3(u_4-x_4)=(-z_3y_3+y_4)x_3+x_1(u_3-x_3)+x_2(u_4-x_4) \label{eqB1}\\
        &z_3v_3-v_4=z_3y_3-y_4\label{eqB2}\\
        &u_3-v_2=x_3-y_2\label{eqB3}\\
        &(-z_4v_3-b+v_1v_3+v_2v_4)u_3-(v_1-z_4)u_4+v_2(a-x_1u_3-x_2u_4)=\nonumber\\
        &(-z_4y_3-b+y_1y_3+y_2y_4)x_3-(y_1-z_4)x_4+y_2(a-x_1x_3-x_2x_4)\label{eqB4}\\
        &z_4v_3-v_1v_3-v_2v_4+v_2x_1=z_4y_3-y_1y_3-y_2y_4+y_2x_1\label{eqB5}\\
        &v_1+v_2x_2=y_1+y_2x_2\label{eqB6}\\
        &-v_3u_3+w_2u_3+u_4-w_1-w_2v_2=-y_3x_3+z_2x_3+x_4-z_1-z_2y_2\label{eqB7}\\
        &(c-w_1z_3-w_2z_4)v_3u_3+w_1v_4u_3-w_2(b-v_1v_3-v_2v_4)u_3-(c-w_1z_3-w_2z_4+w_2v_1)u_4+\nonumber\\
        &(w_1+w_2v_2)(a-x_1u_3-x_2u_4)=(c-z_1z_3-z_2z_4)y_3x_3+z_1y_4x_3-z_2(b-y_1y_3-y_2y_4)x_3-\nonumber\\
        &(c-z_1z_3-z_2z_4+z_2y_1)x_4+(z_1+z_2y_2)(a-x_1x_3-x_2x_4)\label{eqB8}\\
        &v_3-w_2=y_3-z_2\label{eqB9}\\
        &(-c+w_1z_3+w_2z_4)v_3-w_1v_4+w_2(b-v_1v_3-v_2v_4)+(w_1+w_2v_2)x_1=\nonumber\\
        &(-c+z_1z_3+z_2z_4)y_3-z_1y_4+z_2(b-y_1y_3-y_2y_4)+(z_1+z_2y_2)x_1\nonumber\\
        &-w_1z_3-w_2z_4+w_2v_1+(w_1+w_2v_2)x_2=-z_1z_3-z_2z_4+z_2y_1+(z_1+z_2y_2)x_2. \label{eqB11}
\end{align}
\end{subequations}

From \eqref{eqB6}, \eqref{eqB11} and \eqref{eqB9}, it follows that $w_2=z_2$.
From \eqref{eqB9} and \eqref{eqB2}, we obtain $v_3 = y_3$ and $v_4=y_4$, whereas \eqref{eqB6}, \eqref{eqB5} and \eqref{eqB2} imply $v_2=y_2$.
Then, from \eqref{eqB6} and \eqref{eqB3} it follows that $v_1=y_1$, $u_3=x_3$,
 while from \eqref{eqB6}, \eqref{eqB11} and \eqref{eqB9} we obtain $w_2=z_2$.
Now, equations \eqref{eqB9} and \eqref{eqB2} imply $v_3 = y_3$ and $v_4=y_4$, and equations \eqref{eqB6}, \eqref{eqB5} and \eqref{eqB2} imply $v_2=y_2$.
Then, from \eqref{eqB6} and \eqref{eqB3} it follows that $v_1=y_1$, $u_3=x_3$, from
\eqref{eqB1} and \eqref{eqB2} we obtain $u_4=x_4$, and \eqref{eqB6}, \eqref{eqB11} imply $w_1=z_1$.

That is, equation \eqref{trifac-BSQ} implies the trivial solution 
$$u_3 = x_3,\quad u_4 = x_4,\quad v_1 = y_1,\quad v_2 = y_2,\quad v_3 = y_3,\quad v_4 = y_4,\quad w_1 = z_1,\quad w_2 = z_2,$$
therefore, according to Theorem \ref{KP}, map \eqref{NC-BSQ-map} is a Yang--Baxter map.
\end{proof}

\subsection{Squeeze down to the noncommutative Boussinesq lattice equation}
Recall that a quad-graph equation (see Figure \ref{QD-diagram}) is a lattice equation (or system of equations) of the form
\begin{equation}\label{quad-graph}
    \mathcal{Q}(f_{00},f_{10},f_{01},f_{11};a,b)=0,
\end{equation}
where $f=f(n,m)$, $n,m\in\mathbb{Z}$, $f_{ij}=f(n+j,m+j)$, $a,b\in\mathbb{C}$, and $\mathcal{Q}$ is a polynomial affine-linear function, i.e.
$$
\frac{\partial\mathcal{Q}}{\partial f_{ij}}\neq 0,\quad \frac{\partial^2\mathcal{Q}}{\partial f_{ij}^2}=0.
$$

\begin{figure}[ht]
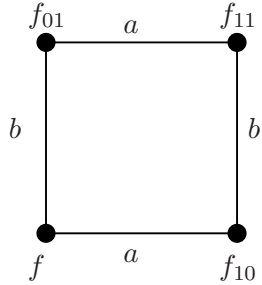

\centertexdraw{ \setunitscale 0.5
\linewd 0.02 \arrowheadtype t:F 
\htext(0 0.5) {\phantom{T}}
\move (-1 -2) \lvec (1 -2) 
\move(-1 -2) \lvec (-1 0) \move(1 -2) \lvec (1 0) \move(-1 0) \lvec(1 0)
\move (1 -2) \fcir f:0.0 r:0.1 \move (-1 -2) \fcir f:0.0 r:0.1
 \move (-1 0) \fcir f:0.0 r:0.1 \move (1 0) \fcir f:0.0 r:0.1  
\htext (-1.2 -2.5) {$f$} \htext (.8 -2.5) {$f_{10}$} \htext (-.2 -2.3) {$a$}
\htext (-1.2 .15) {$f_{01}$} \htext (.8 .15) {$f_{11}$} \htext (-.2 .1) {$a$}
\htext (-1.4 -1) {$b$} \htext (1.1 -1) {$b$}}
\caption{{Quad graph equation.}} \label{bianchi}\label{QD-diagram}
\end{figure}

Equation \eqref{quad-graph} is called integrable if it is equivalent to the following matrix refactorisation problem:
$$
{\rm K}(f_{01},f_{11};a,\lambda){\rm M}(f_{00},f_{01};b,\lambda)={\rm M}(f_{10},f_{11};b,\lambda){\rm K}(f_{00},f_{10};a,\lambda),\quad\text{for all}\quad \lambda\in\mathbb{C},
$$
where ${\rm K}$ and {\rm M} are square matrices. The above equation is called a Lax equation (or representation) of equation \eqref{quad-graph}.

In this section, we aim to construct a noncommutative integrable Boussinesq type lattice system. We shall do that by squeezing Yang--Baxter map \eqref{NC-BSQ-map} down to an integrable lattice equation after a change of variables. 

In order to find the suitable change of variables, we set in \eqref{NC-Lax-BSQ}: $\bm{x}_1=\bm{p}, \bm{x}_2=\bm{q}, \bm{x}_3=\bm{q}_{10}, \bm{x}_4=\bm{r}_{10}$, namely we consider the matrix:
\begin{equation}\label{NC-Lax-BSQ}
{\rm L}(\bm{p}, \bm{q}, \bm{q}_{10}, \bm{r}_{10},a):=
\left(
\begin{matrix}
 -\bm{q}_{10} & 1 & 0 \\
 -\bm{r}_{10} & 0 & 1 \\
a-\bm{p} \bm{q}_{10}-\bm{q} \bm{r}_{10}-\lambda & \bm{p} & \bm{q}
\end{matrix}\right),
\end{equation}
Then, we compare the Lax equation 
\begin{equation}\label{Lax-lattice-BSQ}
    {\rm L}(\bm{p}_{01}, \bm{q}_{01}, \bm{q}_{11}, \bm{r}_{11},a){\rm L}(\bm{p}, \bm{q}, \bm{q}_{01}, \bm{r}_{01},b)={\rm L}(\bm{p}_{10}, \bm{q}_{10}, \bm{q}_{11}, \bm{r}_{11},b){\rm L}(\bm{p}, \bm{q}, \bm{q}_{10}, \bm{r}_{10},a),
\end{equation}
to the matrix refactorisation problem
$$
    {\rm L}(\bm{u}_1, \bm{u}_2, \bm{u}_3, \bm{u}_4,a){\rm L}(\bm{v}_1, \bm{v}_2, \bm{v}_3, \bm{v}_4,b)={\rm L}(\bm{y}_1, \bm{y}_2, \bm{y}_3, \bm{y}_4,b){\rm L}(\bm{x}_1, \bm{x}_2, \bm{x}_3, \bm{x}_4,a),
$$
which generates the Yang--Baxter map \eqref{NC-BSQ-map}. This indicates the suitable change of variables.

We summarise everything in the following. 

\begin{proposition} (Noncommutative Boussinesq lattice equation)
    Yang--Baxter map \eqref{NC-BSQ-map} can be squeezed down to the following noncommutative lattice system
    \begin{subequations}\label{BSQ-quad-G}
\begin{align}
&(\bm{p}+\bm{q}\bm{q}_{11}-\bm{r}_{11})(\bm{p}_{01}-\bm{p}_{10})=(a-b)\bm{q},\label{BSQ-quad-G-a}\\
&(\bm{q}_{01}-\bm{q}_{10})(\bm{p}+\bm{q}\bm{q}_{11}-\bm{r}_{11})=b-a,\label{BSQ-quad-G-b}\\
&(\bm{r}_{01}-\bm{r}_{10})(\bm{p}+\bm{q}\bm{q}_{11}-\bm{r}_{11})=(b-a)\bm{q}_{11}.\label{BSQ-quad-G-c}
\end{align}
\end{subequations}
The above noncommutative Boussinesq type system is integrable with Lax representation:
$$
    {\rm L}(\bm{p}_{01}, \bm{q}_{01}, \bm{q}_{11}, \bm{r}_{11},a){\rm L}(\bm{p}, \bm{q}, \bm{q}_{01}, \bm{r}_{01},b)={\rm L}(\bm{p}_{10}, \bm{q}_{10}, \bm{q}_{11}, \bm{r}_{11},b){\rm L}(\bm{p}, \bm{q}, \bm{q}_{10}, \bm{r}_{10},a),
$$
where ${\rm L}(\bm{p}, \bm{q}, \bm{q}_{10}, \bm{r}_{10},a):=
\left(
\begin{matrix}
 -\bm{q}_{10} & 1 & 0 \\
 -\bm{r}_{10} & 0 & 1 \\
a-\bm{p} \bm{q}_{10}-\bm{q} \bm{r}_{10}-\lambda & \bm{p} & \bm{q}
\end{matrix}\right).$ Finally, system admits the following conservation law:
\begin{equation}\label{consLaws-a}
(\mathcal{T}-1)(\bm{p}_{10}+\bm{q}_{10}\bm{q}-\bm{r})=(\mathcal{S}-1)(\bm{p}_{01}+\bm{q}_{01}\bm{q}-\bm{r}).
\end{equation}
\end{proposition}
\begin{proof}
    We set $\bm{u}_1=\bm{p}_{01}$, $\bm{u}_2=\bm{v}_3=\bm{q}_{01}$, $\bm{u}_3=\bm{y}_3=\bm{q}_{11}$, $\bm{u}_4=\bm{y}_4=\bm{r}_{11}$, $\bm{v}_1=\bm{x}_1=\bm{p}$, $\bm{v}_2=\bm{x}_2=\bm{q}$, $\bm{y}_2=\bm{x}_3=\bm{q}_{10}$, $\bm{v}_4=\bm{r}_{01}$, $\bm{x}_4=\bm{r}_{10}$, $\bm{y}_1=\bm{p}_{10}$.

    Then, equation
    $$
    \bm{u}_1=\bm{y}_1-(a-b) (\bm{y}_4-\bm{x}_1-\bm{x}_2\bm{y}_3)^{-1}\bm{x}_2,
    $$
    from map \eqref{NC-BSQ-map} can be rewritten as equation \eqref{BSQ-quad-G-a}. Moreover, equation
    $$
    \bm{u}_2=\bm{y}_2+(a-b)(\bm{y}_4-\bm{x}_1-\bm{x}_2\bm{y}_3)^{-1}
    $$
    takes the form of equation \eqref{BSQ-quad-G-b}. Additionally, equation
    $$
    \bm{v}_4=\bm{x}_4+(a-b)\bm{y}_3(\bm{y}_4-\bm{x}_1-\bm{x}_2\bm{y}_3)^{-1},
    $$
    can be written as \eqref{BSQ-quad-G-b}. Now, the map \eqref{NC-BSQ-map} has the property that if $\bm{x}_3=\bm{y}_2$, then $\bm{u}_2=\bm{v}_3$. That is, equation $\bm{v}_3=\bm{x}_3+(a-b)(\bm{y}_4-\bm{x}_1-\bm{x}_2\bm{y}_3)^{-1}$ also takes the form of equation \eqref{BSQ-quad-G-b}. The rest equations of map \eqref{NC-BSQ-map}, namely $\bm{u}_3=\bm{y}_3,\bm{u}_4=\bm{y}_4,\bm{v}_1=\bm{x}_1$ and $\bm{v}_2=\bm{x}_2$, are identically satisfied in the new variables.

    Now, from \eqref{BSQ-quad-G-a} we obtain that
    \begin{equation}\label{p10p01}
        \bm{p}_{01}-\bm{p}_{10}=(a-b)(\bm{p}+\bm{q}\bm{q}_{11}-\bm{r}_{11})^{-1}\bm{q}\stackrel{\eqref{BSQ-quad-G-b}}{=}(\bm{q}_{10}-\bm{q}_{01})\bm{q},
    \end{equation}
    while from \eqref{BSQ-quad-G-c} it follows that
    \begin{equation}\label{r10r01}
        \bm{r}_{01}-\bm{r}_{10}=\bm{q}_{11}(b-a)(\bm{p}+\bm{q}\bm{q}_{11}-\bm{r}_{11})^{-1}\stackrel{\eqref{BSQ-quad-G-b}}{=}\bm{q}_{11}(\bm{q}_{01}-\bm{q}_{10}).
    \end{equation}
    From equations \eqref{p10p01} and \eqref{r10r01} we obtain
    $$
    \bm{q}_{11}\bm{q}_{01}-\bm{r}_{01}-\bm{p}_{10}-\bm{q}_{10}\bm{q}=\bm{q}_{11}\bm{q}_{10}-\bm{r}_{10}-\bm{p}_{01}-\bm{q}_{01}\bm{q},
    $$
    which is equivalent to the conservation law \eqref{consLaws-a}.
\end{proof}

System is \eqref{BSQ-quad-G} the fully noncommutative version of the Boussinesq lattice equation \cite{Tongas-Nijhoff}. A Grassmann extension of the latter was constructed in \cite{Sokor-2020}.

\section{NLS type 3-simplex map}\label{sec-NLS}
Recall that a matrix $\rm{M}$ is called a Darboux matrix for operator $\mathcal{L}$ if the latter is covariant under the similarity transformation $\rm{M}\mathcal{L}\rm{M}^{-1}$, namely,
$$
\rm{M}\mathcal{L}(u(x,t),a,\lambda)\rm{M}^{-1}=\rm{M}\mathcal{L}(\tilde{u}(x,t),a,\lambda)\rm{M}^{-1},\quad\text{for all}\quad\lambda\in\mathbb{C}.
$$

Let $(\bm{p}(x,t),\bm{q}(x,t))$ and $(\tilde{\bm{p}}(x,t),\tilde{\bm{q}}(x,t))$ are solutions of the noncommutative NLS system:
$$
\begin{cases}
\bm{p}_t=\frac{1}{2}\bm{p}_{xx}+4\bm{p}\bm{q}\bm{p}\\  
\bm{q}_t=-\frac{1}{2}\bm{q}_{xx}-4\bm{q}\bm{p}\bm{q}.
\end{cases}
$$
A Darboux transformation for the above system was constructed in \cite{Sokor-PaulX}, and it reads
\begin{equation}\label{DT2-10}
{\rm M} =\lambda\begin{pmatrix}
1 & 0 \\
0 & 0
\end{pmatrix}
+
\begin{pmatrix}
a+\bm{p}\tilde{\bm{q}} & \bm{p} \\
\tilde{\bm{p}} & 1
\end{pmatrix},
\end{equation}
where $\bm{p}$ and $\bm{q}$ satisfy the system of differential equations:
$$
\bm{p}_x=2\bm{p}_{10}-2a\bm{p}-2\bm{p}\bm{q}_{10}\bm{p},\quad \bm{q}_{10,x}=2a\bm{q}_{10}+2\bm{q}_{10}\bm{p}\bm{q}_{10}-2\bm{p}.
$$

Now, we set $\bm{x}_1=\bm{p}(x,t)$ and $\bm{x}_2=\tilde{\bm{q}}(x,t)$ in \eqref{DT2-10}, namely we consider the matrix
\begin{equation}\label{DT-NLS}
{\rm M} (\bm{x}_1,\bm{x}_2,a,\lambda)=\begin{pmatrix}
a +\lambda +\bm{x}_1\bm{x}_2 & \bm{x}_1 \\
\bm{x}_2 & 1
\end{pmatrix},
\end{equation}
where its elements belong to a division ring $\bm{x}_i\in\mathcal{R}$, $i=1,2$, and $a,\lambda\in Z(\mathcal{R})$.

A noncommutative Adler--Yamilov type 2-simplex map was constructed in \cite{Sokor-Nikitina} and is equivalent to the matrix refactorisation problem
$$
{\rm M} (\bm{u}_1,\bm{u}_2,a,\lambda){\rm M} (\bm{v}_1,\bm{v}_2,b,\lambda)={\rm M} (\bm{y}_1,\bm{y}_2,b,\lambda){\rm K} (\bm{x}_1,\bm{x}_2,a,\lambda).
$$
Here, starting from the same matrix \eqref{DT-NLS}, we shall construct a noncommutative 3-simplex map via the local Yang--Baxter equation.

Let ${\rm K} (\bm{x}_1,\bm{x}_2,a)={\rm M} (\bm{x}_1,\bm{x}_2,a,0)$. We consider its $3\times 3$ extensions, namely the following matrices
\begin{align}\label{Mij-mat}
   &{\rm K}^3_{12}(\bm{x}_1,\bm{x}_2,a)=\begin{pmatrix} 
a +\bm{x}_1\bm{x}_2 & \bm{x}_1 & 0\\ 
\bm{x}_2 &  1 & 0\\
0 & 0 & 1
\end{pmatrix},\nonumber\\
& {\rm K}^3_{13}(\bm{x}_1,\bm{x}_2,a)= \begin{pmatrix} 
a +\bm{x}_1\bm{x}_2 & 0 &\bm{x}_1\\ 
0 & 1 & 0\\
\bm{x}_2 & 0 & 1
\end{pmatrix},  \\
 &{\rm K}^3_{23}(\bm{x}_1,\bm{x}_2,a)=\begin{pmatrix} 
   1 & 0 & 0 \\
0 &a +\bm{x}_1\bm{x}_2 &\bm{x}_1\\ 
0 &\bm{x}_2 & 1
\end{pmatrix},\nonumber
\end{align}
and substitute them to the local Yang--Baxter equation in order to construct a parametric 3-simplex tetrahedron map. In particular, we have the following.

\begin{proposition} (Nonocommutative NLS map)
    Let $\mathcal{R}$ be a nonocommutative division ring and $Z(\mathcal{R})$ its centre. The map $T_{a,b,c}:\mathcal{R}^6\rightarrow\mathcal{R}^6$, $a,b,c\in Z(\mathcal{R})$, given by
   \begin{subequations}\label{NC-Tet-NLS} 
\begin{align}
\bm{x}_1\mapsto \bm{u}_1 &=\frac{b \bm{x}_1-\bm{y}_1\bm{z}_2}{c};\label{NC-Tet-NLS-a}\\
\bm{x}_2\mapsto \bm{u}_2 &=ac(\bm{z}_1\bm{y}_2(a+\bm{x}_1\bm{x}_2)+(c+\bm{z}_1\bm{z}_2)\bm{x}_2)\cdot\nonumber\\
&\quad\left[abc+(\bm{y}_1(c+\bm{z}_2\bm{z}_1)-b\bm{x}_1\bm{z}_1)(\bm{y}_2(a+\bm{x}_1\bm{x}_2)+\bm{z}_2\bm{x}_2)\right]^{-1},\label{NC-Tet-NLS-b}\\
\bm{y}_1\mapsto \bm{v}_1 &=\frac{\bm{y}_1(c+\bm{z}_2\bm{z}_1)-b\bm{x}_1\bm{z}_1}{ac};\label{NC-Tet-NLS-c}\\
\bm{y}_2\mapsto \bm{v}_2 &=\bm{z}_2\bm{x}_2+\bm{y}_2(a+\bm{x}_1\bm{x}_2);\label{NC-Tet-NLS-d}\\
\bm{z}_1\mapsto \bm{w}_1 &=\bm{z}_1-(\bm{z}_1\bm{y}_2 (a+\bm{x}_1\bm{x}_2)+(c+\bm{z}_1\bm{z}_2)\bm{x}_2)\cdot\nonumber\\
&\quad \left[\bm{y}_2 (a+\bm{x}_1\bm{x}_2)+\bm{z}_2\bm{x}_2+abc(\bm{y}_1(c+\bm{z}_2\bm{z}_1)-b\bm{x}_1\bm{z}_1)^{-1}\right]^{-1};\label{NC-Tet-NLS-e}\\
\bm{z}_2\mapsto \bm{w}_2 &=\bm{z}_2+\bm{y}_2\bm{x}_1,\label{NC-Tet-NLS-f}
\end{align}
\end{subequations}
is a noninvolutive map and it is equivalent with the local Yang--Baxter equation
\begin{equation}\label{local-YB-NLS}
    {\rm K}^3_{12}(\bm{u}_1,\bm{u}_2,a){\rm K}^3_{13}(\bm{v}_1,\bm{v}_2,b){\rm K}^3_{23}(\bm{w}_1,\bm{w}_2,c)={\rm K}^3_{23}(\bm{z}_1,\bm{z}_2,c){\rm K}^3_{13}(\bm{y}_1,\bm{y}_2,b){\rm K}^3_{12}(\bm{x}_1,\bm{x}_2,a),
\end{equation}
where ${\rm K}^3_{ij}$ are $3\times 3$ extensions of matrix
\begin{equation}\label{Lax-NLS}
{\rm K} (\bm{x}_1,\bm{x}_2,a)=\begin{pmatrix}
a +\bm{x}_1\bm{x}_2 & \bm{x}_1 \\
\bm{x}_2 & 1
\end{pmatrix}.
\end{equation}
\end{proposition}

\begin{proof}
    The local Yang--Baxter equation \eqref{local-YB-NLS} is equivalent to the system of polynomial equations:
    \begin{subequations}
        \begin{align}
            &(a+\bm{u}_1 \bm{u}_2)(b+\bm{v}_1\bm{v}_2)=(b+\bm{y}_1\bm{y}_2)(a+\bm{x}_1\bm{x}_2),\label{NCNLS-1}\\
            &(a+\bm{u}_1\bm{u}_2)\bm{v}_1\bm{w}_2+\bm{u}_1(c+\bm{w}_1\bm{w}_2)=(b+\bm{y}_1\bm{y}_2)\bm{x}_1,\label{NCNLS-2}\\
            &(a+\bm{u}_1\bm{u}_2)\bm{v}_1+\bm{u}_1\bm{w}_1=\bm{y}_1,\label{NCNLS-3}\\
            &\bm{u}_2(b+\bm{v}_1\bm{v}_2)=\bm{z}_1\bm{y}_2(a+\bm{x}_1\bm{x}_2)+(c+\bm{z}_1\bm{z}_2)\bm{x}_2,\label{NCNLS-4}\\
            &\bm{u}_2\bm{v}_1\bm{w}_2+\bm{w}_1\bm{w}_2=\bm{z}_1\bm{y}_2\bm{x}_1+\bm{z}_1\bm{z}_2,\label{NCNLS-5}\\
            &\bm{u}_2\bm{v}_1+\bm{w}_1=\bm{z}_1,\label{NCNLS-6}\\
            &\bm{v}_2=\bm{y}_2(a+\bm{x}_1\bm{x}_2)+\bm{z}_2\bm{x}_2,\label{NCNLS-7}\\
            &\bm{w}_2=\bm{y}_2\bm{x}_1+\bm{z}_2.\label{NCNLS-8}
        \end{align}
    \end{subequations}

    Solving \eqref{NCNLS-6} for ``$\bm{u}_2\bm{v}_1$'' and substituting it into equation \eqref{NCNLS-3}, we obtain:
    \begin{equation}\label{NCNLS-3-2}
        a\bm{v}_1+\bm{u}_1\bm{z}_1=\bm{y}_1.
    \end{equation}
    Equation \eqref{NCNLS-1}, with the use of equations \eqref{NCNLS-6} and \eqref{NCNLS-7} is equivalent to \eqref{NC-Tet-NLS-a}. Now, substituting the obtained $\bm{u}_1$ into \eqref{NCNLS-3-2}, we can rewrite the latter equivalently as \eqref{NC-Tet-NLS-c}. Then, substituting into \eqref{NCNLS-4} the obtained $\bm{v}_1$ and also $\bm{v}_2$ from \eqref{NCNLS-7}, after some calculations, we can write equation \eqref{NCNLS-4} equivalently to \eqref{NC-Tet-NLS-b}. Furthermore, equation \eqref{NC-Tet-NLS-e} is equivalent to \eqref{NCNLS-4} after substitution into the latter of the obtained $\bm{u}_2$ and $\bm{v}_1$.

    Finally, we have that
    $$
    \bm{w}_2\circ T_{a,b,c}=\bm{z}_2+\bm{y}_2\bm{x}_1+\frac{1}{c}(\bm{z}_2\bm{x}_2+\bm{y}_2(a+\bm{x}_1\bm{x}_2))(b \bm{x}_1-\bm{y}_1\bm{z}_2)\neq \bm{z}_2,
    $$
    which proves that map \eqref{NC-Tet-NLS} is noninvolutive.
\end{proof}

The above is the noncommutative avatar of the NLS type tetrahedron map constructed in \cite{Sokor-2020-2}. For the commutative version the tetrahedron 3-simplex property could be verified by straightforward substitution which is impossible in the noncommutative case even with the use of programmes for symbolic computation. In what follows we provide a proof via the matrix six-factorisation problem.

\begin{theorem}
    Map \eqref{NC-Tet-NLS} is a noncommutative tetrahedron 3-simplex map.
\end{theorem}
\begin{proof}
Consider the $4\times 4$ extensions of matrix \eqref{Lax-NLS}, namely the following:
{\small
\begin{align}\label{Kij4}
        &{\rm K}^4_{12}(\bm{x}_1,\bm{x}_2,a)=\begin{pmatrix}a +\bm{x}_1\bm{x}_2 & \bm{x}_1 & 0 & 0 \\ \bm{x}_2 & 1 & 0 & 0 \\ 0 & 0 & 1 & 0 \\ 0 & 0 & 0 & 1\end{pmatrix},\quad
    {\rm K}^4_{13}(\bm{x}_1,\bm{x}_2,a)=\begin{pmatrix}a +\bm{x}_1\bm{x}_2 & 0 & \bm{x}_1 & 0 \\ 0 & 1 & 0 & 0 \\ \bm{x}_2 & 0 & 1 & 0 \\ 0 & 0 & 0 & 1\end{pmatrix},\nonumber\\
        &{\rm K}^4_{23}(\bm{x}_1,\bm{x}_2,a)=\begin{pmatrix}1 & 0 & 0 & 0 \\ 0 & a +\bm{x}_1\bm{x}_2 & \bm{x}_1 & 0 \\ 0 & \bm{x}_2 & 1 & 0 \\ 0 & 0 & 0 & 1\end{pmatrix},\quad {\rm K}^4_{14}(\bm{x}_1,\bm{x}_2,a)=\begin{pmatrix}a +\bm{x}_1\bm{x}_2 & 0 & 0 & \bm{x}_1 \\ 0 & 1 & 0 & 0 \\ 0 & 0 & 1 & 0 \\ \bm{x}_2 & 0 & 0 & 1\end{pmatrix}, \\
        & {\rm K}^4_{24}(\bm{x}_1,\bm{x}_2,a)=\begin{pmatrix}1 & 0 & 0 & 0 \\ 0 & a +\bm{x}_1\bm{x}_2 & 0 & \bm{x}_1 \\ 0 & 0 & 1 & 0 \\ 0 & \bm{x}_2 & 0 & 1\end{pmatrix}, \quad
    {\rm K}^4_{34}(\bm{x}_1,\bm{x}_2,a)=\begin{pmatrix}1 & 0 & 0 & 0 \\ 0 & 1 & 0 & 0 \\ 0 & 0 & a +\bm{x}_1\bm{x}_2 & \bm{x}_1 \\ 0 & 0 & \bm{x}_2 & 1\end{pmatrix},
\end{align}}
and substitute them to the Lax equation:
\begin{align*}
&{\rm K}^4_{34}(\bm{t}_1,\bm{t}_2;a_6){\rm K}^4_{24}(\bm{s}_1,\bm{s}_2;a_5){\rm K}^4_{14}(\bm{r}_1,\bm{r}_2;a_4){\rm K}^4_{23}(\bm{z}_1,\bm{z}_2;a_3){\rm K}^4_{13}(\bm{y}_1,\bm{y}_2;a_2){\rm K}^4_{12}(\bm{x}_1,\bm{x}_2;a_1)=\\
&{\rm K}^4_{34}(\tilde{\bm{t}}_1,\tilde{\bm{t}}_2;a_6){\rm K}^4_{24}(\tilde{\bm{s}}_1,\tilde{\bm{s}}_2;a_5){\rm K}^4_{14}(\tilde{\bm{r}}_1,\tilde{\bm{r}}_2;a_4){\rm K}^4_{23}(\tilde{\bm{z}}_1,\tilde{\bm{z}}_2;a_3){\rm K}^4_{13}(\tilde{\bm{y}}_1,\tilde{\bm{y}}_2;a_2){\rm K}^4_{12}(\tilde{\bm{x}}_1,\tilde{\bm{x}}_2;a_1).
\end{align*}

The above matrix six-factorisation problem is equivalent to the system of polynomial equations:{\small
\begin{subequations}\label{Left-6fac-NLS}
\begin{align}
&\tilde{\bm{r}}_1 = \bm{r}_1, \qquad \tilde{\bm{s}}_1 =\bm{s}_1,\qquad  \tilde{\bm{t}}_1=\bm{t}_1,\label{Left-6fac-NLS-a}\\
&(a_4+\tilde{\bm{r}}_1\tilde{\bm{r}}_2)(a_2+\tilde{\bm{y}}_1\tilde{\bm{y}}_2)(a_1+\tilde{\bm{x}}_1\tilde{\bm{x}}_2)=(a_4+\bm{r}_1\bm{r}_2)(a_2+\bm{y}_1\bm{y}_2)(a_1+\bm{x}_1\bm{x}_2),\label{Left-6fac-NLS-b}\\
&(a_4+\tilde{\bm{r}}_1\tilde{\bm{r}}_2)(a_2+\tilde{\bm{y}}_1\tilde{\bm{y}}_2)\tilde{\bm{x}}_1 =(a_4+\bm{r}_1\bm{r}_2)(a_2+\bm{y}_1\bm{y}_2)\bm{x}_1,\label{Left-6fac-NLS-c}\\
&(a_4+\tilde{\bm{r}}_1\tilde{\bm{r}}_2)\tilde{\bm{y}}_1=(a_4+\bm{r}_1\bm{r}_2)\bm{y}_1,\label{Left-6fac-NLS-d}\\
&\tilde{\bm{s}}_1\tilde{\bm{r}}_2(a_2+\tilde{\bm{y}}_1\tilde{\bm{y}}_2)(a_1+\tilde{\bm{x}}_1\tilde{\bm{x}}_2)+(a_5+\tilde{\bm{s}}_1\tilde{\bm{s}}_2)(a_3+\tilde{\bm{z}}_1\tilde{\bm{z}}_2)\tilde{\bm{x}}_2+(a_5+\tilde{\bm{s}}_1\tilde{\bm{s}}_2)\tilde{\bm{z}}_1\tilde{\bm{y}}_2(a_1+\tilde{\bm{x}}_1\tilde{\bm{x}}_2) = \nonumber\\
&\bm{s}_1\bm{r}_2(a_2+\bm{y}_1\bm{y}_2)(a_1+\bm{x}_1\bm{x}_2)+(a_5+\bm{s}_1\bm{s}_2)(a_3+\bm{z}_1\bm{z}_2)\bm{x}_2+(a_5+\bm{s}_1\bm{s}_2)\bm{z}_1\bm{y}_2(a_1+\bm{x}_1\bm{x}_2),\label{Left-6fac-NLS-e}\\
&\tilde{\bm{s}}_1\tilde{\bm{r}}_2(a_2+\tilde{\bm{y}}_1\tilde{\bm{y}}_2)\tilde{\bm{x}}_1+(a_5+\tilde{\bm{s}}_1\tilde{\bm{s}}_2)(a_3+\tilde{\bm{z}}_1\tilde{\bm{z}}_2)+(a_5+\tilde{\bm{s}}_1\tilde{\bm{s}}_2)\tilde{\bm{z}}_1\tilde{\bm{y}}_2\tilde{\bm{x}}_1 = \nonumber\\ 
&\bm{s}_1\bm{r}_2(a_2+\bm{y}_1\bm{y}_2)\bm{x}_1+(a_5+\bm{s}_1\bm{s}_2)(a_3+\bm{z}_1\bm{z}_2)+(a_5+\bm{s}_1\bm{s}_2)\bm{z}_1\bm{y}_2\bm{x}_1,\label{Left-6fac-NLS-f}\\
&\tilde{\bm{t}}_1\tilde{\bm{r}}_2(a_2+\tilde{\bm{y}}_1\tilde{\bm{y}}_2)(a_1+\tilde{\bm{x}}_1\tilde{\bm{x}}_2)+\tilde{\bm{t}}_1\tilde{\bm{s}}_2(a_3+\tilde{\bm{z}}_1\tilde{\bm{z}}_2)\tilde{\bm{x}}_2+(a_6+\tilde{\bm{t}}_1\tilde{\bm{t}}_2)\tilde{\bm{z}}_2\tilde{\bm{x}}_2+(\tilde{\bm{t}}_1\tilde{\bm{s}}_2\tilde{\bm{z}}_1+a_6+\tilde{\bm{t}}_1\tilde{\bm{t}}_2)\tilde{\bm{y}}_2(a_1+\tilde{\bm{x}}_1\tilde{\bm{x}}_2)=\nonumber\\
&\bm{t}_1\bm{r}_2(a_2+\bm{y}_1\bm{y}_2)(a_1+\bm{x}_1\bm{x}_2)+\bm{t}_1\bm{s}_2(a_3+\bm{z}_1\bm{z}_2)\bm{x}_2+(a_6+\bm{t}_1\bm{t}_2)\bm{z}_2\bm{x}_2+(\bm{t}_1\bm{s}_2\bm{z}_1+a_6+\bm{t}_1\bm{t}_2)\bm{y}_2(a_1+\bm{x}_1\bm{x}_2),\label{Left-6fac-NLS-g}\\
&\tilde{\bm{t}}_1\tilde{\bm{r}}_2(a_2+\tilde{\bm{y}}_1\tilde{\bm{y}}_2)\tilde{\bm{x}}_1+\tilde{\bm{t}}_1\tilde{\bm{s}}_2(a_3+\tilde{\bm{z}}_1\tilde{\bm{z}}_2)+(a_6+\tilde{\bm{t}}_1\tilde{\bm{t}}_2)\tilde{\bm{z}}_2+(\tilde{\bm{t}}_1\tilde{\bm{s}}_2\tilde{\bm{z}}_1+a_6+\tilde{\bm{t}}_1\tilde{\bm{t}}_2)\tilde{\bm{y}}_2\tilde{\bm{x}}_1=\nonumber\\
&\bm{t}_1\bm{r}_2(a_2+\bm{y}_1\bm{y}_2)\bm{x}_1+
\bm{t}_1\bm{s}_2(a_3+\bm{z}_1\bm{z}_2)+(a_6+\bm{t}_1\bm{t}_2)\bm{z}_2+(\bm{t}_1\bm{s}_2\bm{z}_1+a_6+\bm{t}_1\bm{t}_2)\bm{y}_2\bm{x}_1,\label{Left-6fac-NLS-h}\\
&\tilde{\bm{t}}_1\tilde{\bm{r}}_2\tilde{\bm{y}}_1+\tilde{\bm{t}}_1\tilde{\bm{s}}_2\tilde{\bm{z}}_1+a_6+\tilde{\bm{t}}_1\tilde{\bm{t}}_2=\bm{t}_1\bm{r}_2\bm{y}_1+\bm{t}_1\bm{s}_2\bm{z}_1+a_6+\bm{t}_1\bm{t}_2,\label{Left-6fac-NLS-i}\\
&\tilde{\bm{r}}_2(a_2+\tilde{\bm{y}}_1\tilde{\bm{y}}_2)(a_1+\tilde{\bm{x}}_1\tilde{\bm{x}}_2)+
\tilde{\bm{s}}_2(a_3+\tilde{\bm{z}}_1\tilde{\bm{z}}_2)\tilde{\bm{x}}_2+\tilde{\bm{t}}_2\tilde{\bm{z}}_2\tilde{\bm{x}}_2+(\tilde{\bm{s}}_2\tilde{\bm{z}}_1+\tilde{\bm{t}}_2)\tilde{\bm{y}}_2(a_1+\tilde{\bm{x}}_1\tilde{\bm{x}}_2)=\nonumber\\
&\bm{r}_2(a_2+\bm{y}_1\bm{y}_2)(a_1+\bm{x}_1\bm{x}_2)+\bm{s}_2(a_3+\bm{z}_1\bm{z}_2)\bm{x}_2+\bm{t}_2\bm{z}_2\bm{x}_2+(\bm{s}_2\bm{z}_1+\bm{t}_2)\bm{y}_2(a_1+\bm{x}_1\bm{x}_2),\label{Left-6fac-NLS-j}\\
&\tilde{\bm{r}}_2(a_2+\tilde{\bm{y}}_1\tilde{\bm{y}}_2)\tilde{\bm{x}}_1+\tilde{\bm{s}}_2(a_3+\tilde{\bm{z}}_1\tilde{\bm{z}}_2)+\tilde{\bm{t}}_2\tilde{\bm{z}}_2+(\tilde{\bm{s}}_2\tilde{\bm{z}}_1+\tilde{\bm{t}}_2)\tilde{\bm{y}}_2\tilde{\bm{x}}_1=\nonumber\\
&\bm{r}_2(a_2+\bm{y}_1\bm{y}_2)\bm{x}_1+\bm{s}_2(a_3+\bm{z}_1\bm{z}_2)+\bm{t}_2\bm{z}_2+(\bm{s}_2\bm{z}_1+\bm{t}_2)\bm{y}_2\bm{x}_1,\label{Left-6fac-NLS-k}\\
&\tilde{\bm{r}}_2\tilde{\bm{y}}_1+\tilde{\bm{s}}_2\tilde{\bm{z}}_1+\tilde{\bm{t}}_2 = \bm{r}_2\bm{y}_1+\bm{s}_2\bm{z}_1+\bm{t}_2,\label{Left-6fac-NLS-l}
\end{align}
\end{subequations}

From equations (\ref{Left-6fac-NLS-b}) and (\ref{Left-6fac-NLS-c}), we obtain 
\begin{equation}
    (\bm{x}_2-\tilde{\bm{x}}_2) = (\tilde{\bm{x}}_1^{-1}-\bm{x}_1^{-1})a_1,\label{NLS-1'},
\end{equation}
while equations (\ref{Left-6fac-NLS-c}) and (\ref{Left-6fac-NLS-d}) imply
\begin{equation}
a_2(\bm{x}_1-\bm{y}_1\tilde{\bm{y}}_1^{-1}\tilde{\bm{x}}_1)+\bm{y}_1(\bm{y}_2\bm{x}_1-\tilde{\bm{y}}_2\tilde{\bm{x}}_1) = 0.\label{NLS-2'}
\end{equation}
Then, from (\ref{Left-6fac-NLS-e})-(\ref{Left-6fac-NLS-f}), (\ref{Left-6fac-NLS-g})-(\ref{Left-6fac-NLS-h}), and (\ref{Left-6fac-NLS-j})-(\ref{Left-6fac-NLS-k}) we obtain
\begin{subequations}
\begin{align}
&a_1a_2\bm{s}_1(\tilde{\bm{r}}_2-\bm{r}_2)+a_1(\bm{s}_1\bm{r}_2\bm{y}_1+(a_5+\bm{s}_1\bm{s}_2)\bm{z}_1)(\tilde{\bm{y}}_2-\bm{y}_2)+(\bm{s}_1\bm{r}_2(a_2+\bm{y}_1\bm{y}_2)\bm{x}_1+\nonumber\\
&+(a_5+\bm{s}_1\bm{s}_2)(a_3+\bm{z}_1\bm{z}_2-\bm{z}_1\bm{y}_2\bm{x}_1))(\bm{x}_2-\tilde{\bm{x}}_2) = 0,\label{NLS-4'} \\
&a_1a_2\bm{t}_1(\tilde{\bm{r}}_2-\bm{r}_2)+a_1(\bm{t}_1\bm{r}_2\bm{y}_1+\bm{t}_1\bm{s}_2\bm{z}_1+a_6+\bm{t}_1\bm{t}_2)(\tilde{\bm{y}}_2-\bm{y}_2) + (\bm{t}_1\bm{r}_2(a_2+\bm{y}_1\bm{y}_2)\bm{x}_1+\nonumber\\
&+\bm{t}_1\bm{s}_2(a_3+\bm{z}_1\bm{z}_2)+(a_6+\bm{t}_1\bm{t}_2)\bm{z}_2+(\bm{t}_1\bm{s}_2\bm{z}_1+a_6+\bm{t}_1\bm{t}_2)\bm{y}_2\bm{x}_1)(\tilde{\bm{x}}_2-\bm{x}_2)=0,\label{NLS-7'} \\
&a_1a_2(\tilde{\bm{r}}_2-\bm{r}_2)+a_1(\bm{r}_2\bm{y}_1+\bm{s}_2\bm{z}_1+\bm{t}_2)(\tilde{\bm{y}}_2-\bm{y}_2) +(\bm{r}_2(a_2+\bm{y}_1\bm{y}_2)\bm{x}_1+\bm{s}_2(a_3+\bm{z}_1\bm{z}_2)+\nonumber\\
&\bm{t}_2\bm{z}_2+(\bm{s}_2\bm{z}_1+\bm{t}_2)\bm{y}_2\bm{x}_1)(\tilde{\bm{x}}_2-\bm{x}_2)=0,\label{NLS-10'}
\end{align}
\end{subequations}
respectively. 

Now, if we substitute the expression $a_1a_2(\tilde{\bm{r}}_2-\bm{r}_2)$  from (\ref{NLS-10'}) into equations \eqref{NLS-4'} and \eqref{NLS-7'}, we obtain
\begin{subequations}
\begin{align}
&a_1(a_5\bm{z}_1-\bm{s}_1\bm{t}_2)(\tilde{\bm{y}}_2-\bm{y}_2)+(a_5(a_3+\bm{z}_1\bm{z}_2-\bm{z}_1\bm{y}_2\bm{x}_1)-\bm{s}_1\bm{t}_2(\bm{z}_2+\bm{y}_2\bm{x}_1))(\tilde{\bm{x}}_2-\bm{x}_2) = 0, \label{NLS-4''} \\
&a_1(\tilde{\bm{y}}_2-\bm{y}_2)+(\bm{z}_2+\bm{y}_2\bm{x}_1)(\tilde{\bm{x}}_2-\bm{x}_2) = 0. \label{NLS-7''}
\end{align}
\end{subequations}
Then, substituting \eqref{NLS-7''} into \eqref{NLS-4''}, it follows that $\tilde{\bm{x}}_2 = \bm{x}_2$.
Therefore, from \eqref{NLS-4''} we obtain $\tilde{\bm{y}}_2 = \bm{y}_2$, whereas \eqref{NLS-4'} implies $\tilde{\bm{r}}_2 = \bm{r}_2$.
Moreover, from \eqref{NLS-1'} and \eqref{NLS-2'} we obtain $\tilde{\bm{x}}_1 = \bm{x}_1$ and $\tilde{\bm{y}}_1 = \bm{y}_1$.
Now, if we substitute \eqref{Left-6fac-NLS-k} into \eqref{Left-6fac-NLS-l}, we obtain  $\tilde{\bm{z}}_2 = \bm{z}_2$, and from \eqref{Left-6fac-NLS-k} it follows that $\tilde{\bm{s}}_2 = \bm{s}_2$. Finally, equations \eqref{Left-6fac-NLS-f} and \eqref{Left-6fac-NLS-l} imply $\tilde{\bm{z}}_1 = \bm{z}_1$ and $\tilde{\bm{t}}_2 = \bm{t}_2$.

From the above it follows that 
$$
\tilde{\bm{x}}_i=\bm{x}_i,\quad \tilde{\bm{y}}_i=\bm{y}_i,\quad \tilde{\bm{z}}_i=\bm{z}_i,\quad \tilde{\bm{r}}_i=\bm{r}_i, \quad \tilde{\bm{s}}_i=\bm{s}_i,\quad\text{and}\quad \tilde{\bm{t}}_i=\bm{t}_i,\quad i=1,2.
$$
Therefore, according to Theorem \ref{six-factorisation}, map \eqref{NC-Tet-NLS} is a tetrahedron 3-simplex map.
}
\end{proof}

\section{Conclusions}
In this paper we constructed commutative and noncommutative 2-simplex and 3-simplex maps related to the lattice Boussinesq and the NLS equation. 

In particular we constructed an eight-dimensional parametric Yang--Baxter map \eqref{YB-2-BSQ} which admits 4 functionally independent invariants. This indicates integrability. However, in order to claim Liouville integrability of map \eqref{YB-2-BSQ}, one must we need a Poisson bracket with respect to which the integrals \eqref{BSQ-ints-a}--\eqref{BSQ-ints-d} are in involution. It is worth noting that the first integrals \eqref{BSQ-ints-a} and \eqref{BSQ-ints-b} are Casimir functions for the Poisson bracket:
$$
\{x_1,x_3\}=\{x_2,x_4\}=\{x_3,x_4\}=\{x_3,y_1\}=\{y_1,y_3\}=\{y_2,y_4\}=\{y_3,y_4\}=\{y_3,x_1\}=1,
$$
and all the rest $\{x_i,x_j\}, \{x_i,y_i\}, \{y_i,y_j\}$ are 0. Also, the first integrals \eqref{BSQ-ints-a} and \eqref{BSQ-ints-b} are in involution with respect to the above Poisson bracket. However, the rank of the associated Poisson matrix is 6, thus we need on more integral in order to claim Liouville integrability. 

Furthermore, we constructed an eight-dimensional noncommutative Boussinesq type map \eqref{NC-BSQ-map} and we showed that it satisfies the Yang--Baxter equation. After a change of variables indicated by the Lax equation itself we squeezed down this Boussinesq type Yang--Baxter map to the to an integrable nonommmutative Boussinesq lattice system.

The Boussinesq type maps \eqref{YB-2-BSQ} and \eqref{NC-BSQ-map} were constructed using the correspondence approach \cite{Igonin-Sokor}. The choice of equations to supplement the associated correspondences was made in order to: i. construct an integrable Yang--Baxter map (with the ammount of first integrals to be half the dimension of the map) for the case of map \eqref{YB-2-BSQ}, and ii. construct a map which can be squeezed down to the integrable lattice Boussinesq system \eqref{NC-BSQ-map}.

Finally, we employed a Darboux transformation for the noncommutative NLS equation and constructed a six-dimensional noninvolutive NLS type map on and arbitrary noncommutative division ring that is a noncommutative version of the NLS parametric Yang--Baxter map which appeared in \cite{Sokor-2020-2}. We showed that this map solves the Zamolodchikov tetrahedron equation.

The commutative version of map \eqref{NC-Tet-NLS} was derived in \cite{Sokor-2020-2} as restriction of map $
(x_1,x_2,X,y_1,y_2,Y,z_1,z_2,Z)\overset{T}{\longrightarrow} (u_1,u_2,U,v_1,v_2,V,w_1,w_2,W)
$, given by:
{\small
\begin{align}
x_1\mapsto u_1 &=\frac{x_1(y_1y_2-Y)+y_1z_2}{z_1z_2-Z},\nonumber\\
x_2\mapsto u_2 &=\frac{(x_1x_2-X)(y_2z_1X+x_2Z)(z_1z_2-Z)}{y_1y_2z_1(x_1y_2+z_2)X-(x_1y_2z_1+z_1z_2-Z)XY+x_2[y_1z_2+x_1(y_1y_2-Y)]Z},\nonumber\\
X\mapsto U &=\frac{(x_1x_2-X)(y_1y_2-Y)X}{y_1y_2z_1(x_1y_2+z_2)X-(x_1y_2z_1+z_1z_2-Z)XY+x_2[y_1z_2+x_1(y_1y_2-Y)]Z},\label{Tet-NLS-9D}\\
y_1\mapsto v_1 &=\frac{x_1z_1(y_1y_2-Y)+y_1Z}{(x_1x_2-X)(z_1z_2-Z)},\nonumber\\
y_2\mapsto v_2 &=x_2z_2+y_2X,\nonumber\\
Y\mapsto V &=\frac{y_1y_2z_1(x_1y_2+z_2)X-(x_1y_2z_1+z_1z_2-Z)XY+x_2[y_1z_2+x_1(y_1y_2-Y)]Z}{(x_1x_2-X)(z_1z_2-Z)},\nonumber\\
z_1\mapsto w_1 &=\frac{[x_2y_1Z-z_1(y_1y_2-Y)X](z_1z_2-Z)}{y_1y_2z_1(x_1y_2+z_2)X-(x_1y_2z_1+z_1z_2-Z)XY+x_2[y_1z_2+x_1(y_1y_2-Y)]Z},\nonumber\\
z_2\mapsto w_2 &=x_1y_2+z_2,\nonumber\\
Z\mapsto W &=\frac{(x_1x_2-X)(z_1z_2-Z)YZ}{y_1y_2z_1(x_1y_2+z_2)X-(x_1y_2z_1+z_1z_2-Z)XY+x_2[y_1z_2+x_1(y_1y_2-Y)]Z},\nonumber
\end{align}}
on invariant leaves, namely the level sets of invariants $I_1=X-x_1x_2$, $I_2=Y-y_1y_2$ and $I_3=Z-z_1z_2$. However, here we notice that map \eqref{Tet-NLS-9D} also admits the invariants $I_4=XY$ and $I_5=YZ$, which means that other solutions to the parametric 3-simplex equation may exist that are restrictions map \eqref{Tet-NLS-9D} on other invariant leaves.

The results of this paper can be extended in the following ways. 

\begin{itemize}
    \item Study the Liouville integrability of all the maps constructed in this paper. For the noncommutative ones, more first integrals are needed.
    \item Study the solutions of system \eqref{BSQ-quad-G}. For that, one may construct B\"acklund transformations as in \cite{FKRX}.
    \item Find other parametric restrictions of map \eqref{Tet-NLS-9D} and their quantisations.
\end{itemize}

\section*{Acknowledgements}
This work was funded by the Russian Science Foundation project No. 20-71-10110 (https://rscf.ru/en/project/23-71-50012/).
Part of this  work,  namely the proofs of Theorems 3.3, 4.4 and 4.6, was carried out within the framework of a development programme for the Regional Scientific and Educational Mathematical Centre of the P.G. Demidov Yaroslavl State University with financial support from the Ministry of Science and Higher Education of the Russian Federation (Agreement on provision of subsidy from the federal budget No. 075-02-2024-1442).


\begin{thebibliography}{100}
\bibitem{Adler}
{V. Adler.} {Recuttings of polygons.} {Funktsional. Anal. i Prilozhen. 27 79--82 (1993).}

\bibitem{Giota-Miky}
{P. Adamopoulou and G. Papamikos,} {Entwining Yang--Baxter maps over Grassmann algebras} {Physica D: Nonlinear Phenomena \textbf{472}} {134469} (2025).

\bibitem{Clarkson}
{P.A. Clarkson and E. Dowie,} {Rational Solutions of the Boussinesq Equation and Applications to Rogue Waves.} {Transactions of Mathematics and Its Applications, 1, tnx003. Transactions of Mathematics and Its Applications, 1, tnx003.} (2017).

\bibitem{Dimakis-Hoissen-2015}
{A. Dimakis and F. M\"uller-Hoissen,} {Simplex and polygon equations,} {SIGMA 11 (2015)} {042}.

\bibitem{Doliwa-2014}
{A. Doliwa.} {Non-Commutative Rational Yang–Baxter Maps.} {Lett. Math. Phys. \textbf{104}} {(2014)} {299--309}.

\bibitem{Doliwa-Kashaev}
{A. Doliwa, R.M. Kashaev.} {Non-commutative bi-rational maps satisfying Zamolodchikov equation equation, and Desargues lattices,} {J. Math. Phys.} {61 (2020)} {092704}.

\bibitem{FKRX}
{X. Fisenko, S. Konstantinou-Rizos, and P. Xenitidis.} 
{A discrete Darboux-Lax scheme for integrable difference equations} 
{\em Chaos, Solitons and Fractals}  {\textbf{158}} {112059} (2022).

\bibitem{GKM}
{G.G. Grahovski, S. Konstantinou-Rizos, A.V. Mikhailov,} {Grassmann extensions of Yang--Baxter maps,} {J. Phys. A} {49 (2016)} {145202}.

\bibitem{Igonin-Sokor}
{S. Igonin and S. Konstantinou-Rizos.
Local Yang--Baxter correspondences and set-theoretical solutions to the Zamolodchikov tetrahedron equation. 
\emph{J. Phys. A: Math. Theor. \textbf{56}} 275202 (2023).}

\bibitem{Kashaev-Sergeev} 
{R.M. Kashaev, I.G. Koperanov, and S.M. Sergeev,} {Functional Tetrahedron Equation,} {Theor. Math. Phys.} {117 (1998)} {370--384}.

\bibitem{Kassotakis-2023}
{P. Kassotakis,} {Non-Abelian hierarchies of compatible maps,
associated integrable difference systems and Yang--Baxter maps,} {Nonlinearity} {36} (2023) {2514}.

\bibitem{Kassotakis-Kouloukas}
{P. Kassotakis and T. Kouloukas,} {On non-abelian quadrirational Yang--Baxter maps,} {J. Phys. A: Math. Theor.} {55} (2022) {175203}.

\bibitem{Kassotakis-Kouloukas-Maciej}
{P. Kassotakis, T. Kouloukas and M. Nieszporski.} {Non-Abelian elastic collisions, associated difference systems of equations and discrete analytic functions.} {arXiv:2412.03543} (2024).

\bibitem{Pavlos-Maciej-2}
{P. Kassotakis, M. Nieszporski,} {On non-multiaffine consistent-around-the-cube lattice equations,} {Phys. Lett. A 376 (2012)} {3135--3140}.

\bibitem{Kassotakis-Tetrahedron}
{P. Kassotakis, M. Nieszporski, V. Papageorgiou, and A. Tongas,} {Tetrahedron maps and symmetries of three dimensional integrable discrete equations,} {J. Math. Phys. 60 (2019)} {123503}.

\bibitem{Sokor-2020}
{S. Konstantinou-Rizos,}
{On the 3D consistency of a Grassmann extended lattice Boussinesq system.} 
{\em Nuclear Phys. B } {\textbf{951}} (2020) {114878}.

\bibitem{Sokor-2020-2}
{S. Konstantinou-Rizos,}
{Nonlinear Schr\"odinger type tetrahedron maps.} 
{\em Nuclear Phys. B } {\textbf{960}} (2020) {115207}.

\bibitem{Sokor-2022}
{S. Konstantinou-Rizos.}
{Noncommutative solutions to Zamolodchikov's tetrahedron equation and matrix six-factorisation problems.} 
{Physica D: Nonlinear Phenomena} {\textbf{440}} (2022) {133466} 

\bibitem{Sokor-Kouloukas}
{S. Konstantinou-Rizos and T. Kouloukas,} {A noncommutative discrete potential {K}d{V} lift} {J. Math. Phys.} {\textbf{59}} (2018) {063506}.

\bibitem{Sokor-Sasha}
{S. Konstantinou-Rizos and A.V. Mikhailov.}
{Darboux transformations, finite reduction groups and related Yang--Baxter maps.} {J. Phys. A: Math. Theor.} {\textbf{46}} (2013) {425201}

\bibitem{Sokor-Sasha-2016}
{S. Konstantinou-Rizos and A.V. Mikhailov.}
{Anticommutative extension of the Adler map.} {J. Phys. A: Math. Theor.} {\textbf{49}} (2016) {30LT03}

\bibitem{Sokor-Nikitina}
{S. Konstantinou-Rizos and A.A. Nikitina.}
{Yang–Baxter maps of KdV, NLS and DNLS type on division rings.} {Physica D: Nonlinear Phenomena} {\textbf{465}} (2024) {134213}

\bibitem{Sokor-Pap}
{S. Konstantinou-Rizos and G. Papamikos.}
{Entwining Yang--Baxter maps related to NLS type equations.} {J. Phys. A: Math. Theor.} {\textbf{52}} (2019) {485201}.

\bibitem{Sokor-PaulX}
{S. Konstantinou-Rizos and P. Xenitidis.}
{Integrable discretisations of noncommutative NLS equation.} {(in preparation)} (2025).

\bibitem{Kouloukas}
{T.E. Kouloukas and V.G. Papageorgiou} {Yang--Baxter maps with
  first-degree-polynomial {$2\times 2$} Lax matrices} {J. Phys. A} {\textbf{42}} {404012} (2009).

  \bibitem{Korepanov}
{I.G. Korepanov,} {Algebraic integrable dynamical systems, $2+1$-dimensional models in wholly discrete space-time, and inhomogeneous models in 2-dimensional statistical physics,} {(1995) solv-int/9506003.}

\bibitem{MPW}
{A.V. Mikhailov, G. Papamikos and J.P. Wang.}
Darboux transformation for the vector sine-Gordon equation and integrable equations on a sphere. {Lett. Math. Phys.} {\textbf{106}} (2016) {973--996}.

\bibitem{Nijhoff}
{J.M. Maillet, F. Nijhoff,} {The tetrahedron equation and the four-simplex equation,} {Phys. Lett. A 134 (1989)} {221--228}.

\bibitem{Pap-Tongas-Veselov}
{V.G. Papageorgiou, A.G. Tongas, A.P. Veselov.} {Yang--Baxter maps and symmetries of integrable equations on quad-graphs,} {J. Math. Phys. 47 (2006)} {083502}.

\bibitem{Sergeev}
{S.M. Sergeev,} {Solutions of the Functional Tetrahedron Equation Connected with the Local Yang--Baxter Equation for the Ferro-Electric Condition,} {Lett. Math. Phys. 45 (1998)}  {113--119}.

\bibitem{Veselov2}
{Y. Suris, A. Veselov,} {Lax matrices for Yang--Baxter maps} {J. Nonlinear Math. Phys. 10 (2003)}  {223--230}.

\bibitem{Tongas-Nijhoff}
{A. Tongas and F. Nijhoff.} {The Boussinesq integrable system. Compatible lattice and continuum structures} {\em Glasgow Math. J.} {47} {205--219} (2005).
\end{thebibliography}
\end{document}